\documentclass[11pt]{article}

\usepackage{amsthm}
\usepackage{graphicx} 
\usepackage{array} 

\usepackage{amsmath, amssymb, amsfonts, verbatim}
\usepackage{hyphenat,epsfig,multirow}

\usepackage[usenames,dvipsnames]{xcolor}
\usepackage[ruled]{algorithm2e}

\usepackage{tcolorbox}
\tcbuselibrary{skins,breakable}
\tcbset{enhanced jigsaw}

\definecolor{DarkRed}{rgb}{0.5,0.1,0.1}
\definecolor{DarkBlue}{rgb}{0.1,0.1,0.5}

\usepackage[linktocpage=true,
	pagebackref=true,colorlinks,
	linkcolor=DarkRed,citecolor=ForestGreen,
	bookmarks,bookmarksopen,bookmarksnumbered]
	{hyperref}

\usepackage{bm}
\usepackage{url}
\usepackage{xspace}
\usepackage[mathscr]{euscript}

\usepackage{mdframed}

\usepackage[noend]{algpseudocode}
\makeatletter
\def\BState{\State\hskip-\ALG@thistlm}
\makeatother

\usepackage{cite}
\usepackage{enumitem}

\usepackage[margin=1in]{geometry}

\newtheorem{theorem}{Theorem}
\newtheorem{lemma}{Lemma}[section]

\newtheorem{claim}[lemma]{Claim}

\newtheorem{definition}[lemma]{Definition}

\newtheorem{assumption}{Assumption}

\newtheorem*{claim*}{Claim}
\newtheorem*{proposition*}{Proposition}
\newtheorem*{lemma*}{Lemma}
\newtheorem*{problem*}{Problem}

\newtheorem*{mdresult}{Main Result}
\newenvironment{result}{\begin{mdframed}[backgroundcolor=lightgray!40,topline=false,rightline=false,leftline=false,bottomline=false,innertopmargin=2pt]\begin{mdresult}}{\end{mdresult}\end{mdframed}}

\newtheorem{mdinvariant}[lemma]{Lemma}
\newenvironment{Lemma}{\begin{mdframed}[hidealllines=false,backgroundcolor=gray!10,innertopmargin=5pt]\begin{mdinvariant}}{\end{mdinvariant}\end{mdframed}}

\usepackage{tikz}
\usetikzlibrary{arrows}
\usetikzlibrary{arrows.meta}
\usetikzlibrary{shapes}
\usetikzlibrary{backgrounds}
\usetikzlibrary{positioning}
\usetikzlibrary{decorations.markings}
\usetikzlibrary{patterns}
\usetikzlibrary{calc}
\usetikzlibrary{fit}
\usetikzlibrary{snakes}

\usepackage{caption}
\usepackage{subcaption}

\allowdisplaybreaks

\DeclareMathOperator*{\argmax}{arg\,max}

\renewcommand{\qed}{\nobreak \ifvmode \relax \else
      \ifdim\lastskip<1.5em \hskip-\lastskip
      \hskip1.5em plus0em minus0.5em \fi \nobreak
      \vrule height0.75em width0.5em depth0.25em\fi}

\newcommand{\Qed}[1]{\ensuremath{\qed_{\textnormal{~#1}}}}

\usepackage[compact]{titlesec}
\newcommand{\eps}{\ensuremath{\varepsilon}}

\newcommand{\bracket}[1]{\left[#1\right]}
\newcommand{\paren}[1]{\ensuremath{\left(#1\right)}\xspace}
\newcommand{\card}[1]{\left\vert{#1}\right\vert}

\newcommand{\IR}{\ensuremath{\mathbb{R}}}

\newcommand{\ceil}[1]{{\left\lceil{#1}\right\rceil}}

\newcommand{\expect}[1]{\Exp\bracket{#1}}

\newcommand{\set}[1]{\ensuremath{\left\{ #1 \right\}}}
\newcommand{\poly}{\mbox{\rm poly}}

\newcommand{\OPT}{\ensuremath{\mbox{\sc opt}}\xspace}
\newcommand{\opt}{\textnormal{\ensuremath{\mbox{opt}}}\xspace}

\newcommand{\alg}{\ensuremath{\mathcal{A}}\xspace}

\DeclareMathOperator*{\Exp}{\ensuremath{{\mathbb{E}}}}

\newenvironment{tbox}{\begin{tcolorbox}[
		enlarge top by=5pt,
		enlarge bottom by=5pt,
		 breakable,
		 boxsep=0pt,
                  left=4pt,
                  right=4pt,
                  top=10pt,
                  arc=0pt,
                  boxrule=1pt,toprule=1pt,
                  colback=white
                  ]
	}
{\end{tcolorbox}}

\newcommand{\mech}{\ensuremath{\textnormal{\texttt{PriceLearningMechanism}}}\xspace}
\newcommand{\spmech}{\ensuremath{\textnormal{\texttt{FixedPriceAuction}}}\xspace}

\newcommand{\pupdate}{\ensuremath{\textnormal{\texttt{PriceUpdate}}}\xspace}

\newcommand{\bq}{\bm{q}}

\newcommand{\bqi}[1]{\bq^{(#1)}}

\newcommand{\Ai}[1]{\ensuremath{A^{(#1)}}}

\renewcommand{\AA}{\ensuremath{\mathcal{A}}}

\newcommand{\price}{\ensuremath{p}}
\newcommand{\bprice}{\bm{\price}}
\newcommand{\bpricei}[1]{\ensuremath{\bprice^{(#1)}}}

\newcommand{\jstar}{j^{\star}}

\newcommand{\val}[1]{\ensuremath{\textnormal{\textsf{val}}(#1)}}

\newcommand{\Ci}[1]{C^{(#1)}}

\newcommand{\Ms}{M^*}

\newcommand{\Os}{O^{D}}

\newcommand{\barA}{\overline{A}}

\newcommand{\VV}{\ensuremath{\mathcal{V}}}

\renewcommand{\opt}{\ensuremath{\textnormal{\textsf{OPT}}}}
\renewcommand{\OPT}{\ensuremath{\textnormal{\textsf{OPT}}}}

\newcommand{\minprice}{\psi_{\min}}
\newcommand{\maxprice}{\psi_{\max}}
\newcommand{\ratioprice}{\Psi}

\newcommand{\binprice}{\ensuremath{\textnormal{\textsf{price}}}}
\newcommand{\bins}{\ensuremath{\textnormal{\textsf{bins}}}}

\newcommand{\TT}{\ensuremath{\mathcal{T}}}

\newcommand{\TTeven}{\TT^e}
\newcommand{\TTodd}{\TT^o}

\newcommand{\TTstar}{\TT^{\star}}
\newcommand{\partition}{\ensuremath{\textnormal{\texttt{Partition}}}}

\newcommand{\pricei}[1]{\price^{(#1)}}

\newcommand{\qi}[1]{q^{(#1)}}

\renewcommand{\OE}{\ensuremath{OE}}

\newcommand{\Di}[1]{\ensuremath{D^{(#1)}}}

\newcommand{\Ostar}{O^{\star}}
\newcommand{\bqstar}{\bq^{\star}}

\title{Improved Truthful Mechanisms for Combinatorial Auctions \\
with Submodular Bidders}
\author{Sepehr Assadi\footnote{Department of Computer Science, Rutgers University. Part of this work was done while the author was a postdoctoral researcher at Princeton University and was supported in part by the Simons Collaboration on Algorithms and Geometry. 
Email: \texttt{sepehr.assadi@rutgers.edu}.} 
\and Sahil Singla\footnote{Department of Computer Science, Princeton University, and Institute for Advanced Study. Supported in part by the Schmidt Foundation.  Email: \texttt{singla@cs.princeton.edu}.} 
}

\date{}

\begin{document}
\maketitle

\pagenumbering{roman}

\begin{abstract}

A longstanding open problem in Algorithmic Mechanism Design is to design computationally-efficient truthful mechanisms for (approximately) maximizing welfare in combinatorial auctions with submodular bidders. 
The first such mechanism was obtained by Dobzinski, Nisan, and Schapira~[STOC'06] who gave an $O(\log^2{m})$-approximation where $m$ is the number of items. This problem has been studied extensively 
since, culminating in an $O(\sqrt{\log{m}})$-approximation mechanism by Dobzinski~[STOC'16].   

\smallskip

We present a computationally-efficient truthful mechanism with approximation ratio that improves upon the state-of-the-art by  an exponential factor. In particular, our mechanism achieves an $O((\log\log{m})^3)$-approximation in expectation, uses only $O(n)$ demand queries, and has universal truthfulness guarantee. This settles an open question of Dobzinski on whether $\Theta(\sqrt{\log{m}})$ is the best approximation ratio in this setting in negative.

\end{abstract}
\clearpage

\setcounter{tocdepth}{3}
\tableofcontents

\clearpage

\pagenumbering{arabic}
\setcounter{page}{1}


\section{Introduction}\label{sec:intro}
In a combinatorial auction, $m$ items are to be allocated between $n$ bidders. 
Each bidder $i$ has a valuation function $v_i$ that describes their value $v_i(S)$ for every bundle $S$ of items. The goal is to design a mechanism that finds an allocation $A$ of items that maximizes the \emph{social welfare}, which is defined 
as $\val{A} := \sum_{i} v_i(A_i)$ where $A_i$ is the bundle allocated to bidder $i$. 
For a mechanism to be feasible, it needs to be \emph{computationally-efficient}, i.e., run in $\poly(m,n)$ time given access to certain queries to valuation functions, namely value queries and demand queries (see Section~\ref{sec:prelim} for definitions). 
Mechanisms should also take into account the strategic behavior of the bidders. 
A mechanism in which the dominant strategy of each bidder is to reveal their true valuation in response to given queries is called
\emph{truthful}. For randomized mechanisms, we consider \emph{universally truthful} mechanisms which are distributions over truthful mechanisms 
(this is a stronger guarantee than truthful-in-expectation considered also in the literature, e.g.~\cite{LaviS05,DughmiRY11}; see Appendix~\ref{app:truthful}).

A ``paradigmatic''~\cite{DobzinskiNS06,AbrahamBDR12,FotakisKV17}, ``central''~\cite{MualemN08,DughmiV11}, and ``arguably the most important''~\cite{Dobzinski07} problem in Algorithmic Mechanism Design is to design 
mechanisms for combinatorial auctions that are both {computationally-efficient} and {truthful}. 
At the root of this problem is the question of whether there is an inherent clash between computational-efficiency and truthfulness. On one hand, 
the celebrated VCG mechanism of Vickrey-Clarke-Groves \cite{Vickrey61,Clarke71,Groves73} is a truthful mechanism for this problem that returns the welfare maximizing allocation. Alas, this mechanism requires  finding the welfare maximizing allocation {exactly}, which is not possible in $\poly(m,n)$ time for most classes of valuations. On the other hand, from a purely algorithmic point of view, constant factor {approximation} algorithms exist for many interesting classes of valuations, but they are no longer truthful. 

A particular case of this problem that has received significant attention is when  the valuation functions of all the bidders are \emph{submodular}\footnote{A valuation function $v$ is submodular iff $v(S \cup T) + v(S \cap T) \leq v(S) + v(T)$ for all $S$ and $T$; see also Section~\ref{sec:valuations}.} (see, 
e.g.~\cite{DobzinskiNS06,LehmannLN06,DobzinskiS06,DobzinskiV12,Dobzinski07,FeigeV10,Dobzinski11,KrystaV12,DobzinskiV13,Dobzinski16} and references therein). 
There is no poly-time algorithm for finding the optimal allocation of submodular bidders~\cite{MirrokniSV08,FeigeV10,DobzinskiV13} and thus VCG mechanism is not computationally-efficient here. 
On the other hand, by using only value queries, a simple greedy algorithm can achieve a $2$-approximation~\cite{LehmannLN06} and this can be further improved to $(\frac{e}{e-1})$-approximation~\cite{Vondrak08}, 
and even slightly better~\cite{FeigeV06} by using demand queries. This leads to one of the earliest and the most basic questions in Algorithmic Mechanism Design:
\vspace{-4pt}
\begin{quote}
	\emph{How closely can the approximation ratio of truthful mechanisms for submodular bidders match what is possible from an algorithmic point of view that ignore strategic behavior?}
\end{quote}
\vspace{-4pt}

Already more than a decade ago, Dobzinski, Nisan, and Schapira~\cite{DobzinskiNS05} gave the first non-trivial answer to this question by designing an $O(\sqrt{m})$-approximation mechanism. 
This approximation ratio was soon after exponentially improved by the same 
authors~\cite{DobzinskiNS06} to $O(\log^2{m})$, which in turn was improved to $O(\log{m}\log\log{m})$ by Dobzinski~\cite{Dobzinski07}, and then to $O(\log{m})$ by Krysta and V{\"o}cking~\cite{KrystaV12}. 
Breaking this logarithmic barrier remained elusive until a recent breakthrough of Dobzinski~\cite{Dobzinski16} that achieved an $O(\sqrt{\log m})$ approximation.  

\paragraph{Our Result.} We give an \emph{exponential} factor improvement over this $\Theta(\sqrt{\log{m}})$ approximation mechanism of~\cite{Dobzinski16}, proving the following result. 
\begin{result}
	There exists a  universally truthful mechanism for combinatorial auctions with submodular valuations that achieves an  approximation ratio of $O((\log\log{m})^3)$ to the social welfare in expectation using polynomial number of value and demand queries. 
\end{result}

We shall note that our mechanism (as well as all previous ones in~\cite{DobzinskiNS06,Dobzinski07,KrystaV12,Dobzinski16}) actually works for the much broader class of \emph{XOS} valuations (see Section~\ref{sec:prelim} for definition). 
Our result reduces the gap between the approximation ratio of truthful mechanisms vs algorithms for submodular and XOS bidders significantly, namely, from $\poly{(\log{(m))}}$ in prior work to $\poly{(\log\log{(m)})}$.  

Similar to~\cite{Dobzinski16}, our result implies a $\poly(m,n)$ time algorithm with \emph{explicit access} to valuations, when valuations are \emph{budget additive}, i.e., for
every $S$, $v(S) = \min(b,\sum_{j \in S}v(\set{j}))$ for some fixed $b$. These valuations have been studied extensively in the past (see, e.g.~\cite{AndelmanM04,ChakrabartyG08,Dobzinski16}) and a simple reduction from Knapsack shows
that it is NP-hard to compute a demand query for these valuations. Yet, similar to~\cite{Dobzinski16}, our mechanism uses demand queries of a very specific form, and these can be computed in poly-time. We omit the details here and 
instead refer the reader to~\cite[Section 6]{Dobzinski16}. 

\paragraph{Our Techniques.} 
All previous work on this problem~\cite{DobzinskiNS06,Dobzinski07,KrystaV12,Dobzinski16}, at their core, relied on the following key observation: to design truthful mechanisms for submodular or XOS bidders, 
``all'' we need is to find ``good'' estimates of the \emph{item prices} in an optimal allocation; the rest can be handled by a simple \emph{fixed-price auction} using these prices. We also use this observation  but depart from prior work
in the following key conceptual way. Previous work mainly aimed to learn coarse-grained ``statistics'' about the prices, say, the range they should belong to~\cite{DobzinskiNS06,Dobzinski07}, and used these statistics to
``guess'' a small number of good prices (e.g., $O(1)$ prices in~\cite{DobzinskiNS06,Dobzinski07}, and $O(\sqrt{\log{m}})$ in~\cite{Dobzinski16}), whereas we instead strive to ``learn'' the entire price vector of items in a fine-grained way (at least for a large fraction of items). This fine-grained view is the key factor that allows us 
to get much more accurate prices and ultimately leads to the exponentially improved performance of our mechanism. 

A cornerstone of our approach is a ``learning process'' which starts with a simple guess of item prices and \emph{iteratively} refine this guess until it converges to suitable prices for different items. 
Each iteration of this process involves running \emph{several} fixed-price auctions with the prices learned so far and use the resulting allocations to refine our learned prices further. The key to the analysis of 
this mechanism is the ``Learnable-Or-Allocatable Lemma'' (Lemma~\ref{lem:main}): Roughly speaking, we prove that in each iteration of this process, 
we can either refine our learned prices significantly (Learnable), or the fixed-price auction with the currently learned prices already gets a high-welfare allocation (Allocatable).  Thus, after a \emph{few} iterations, the resulting 
prices have been refined enough to allow for a high-welfare allocation. 
One ingredient in the proof of this lemma is 
an interesting property of fixed-price auctions that stems from their greedy nature: if we run a fixed-price auction with a \emph{random ordering} of bidders, 
either we obtain a high-welfare allocation or we sell almost all items (most likely to wrong bidders). Such a property was first proved (in a similar but not identical form) by Dobzinski~\cite{Dobzinski16} and is closely related to other similar results 
about greedy algorithms for maximum matching~\cite{KonradMM12}, matroid intersection~\cite{GS-IPCO17}, and constrained submodular maximization~\cite{Norouzi-FardTMZ18}. 

\paragraph{Further related work.} The gap between the approximation ratio of truthful mechanisms and general algorithms has been studied from numerous angles in the literature. It is known that algorithms that use only $\poly(m,n)$ many value queries, 
or are poly-time in the  input representation (for succinctly representable valuations) can  achieve only $m^{\Omega(1)}$-approximation~\cite{PapadimitriouSS08,Dobzinski11,DughmiV11,DobzinskiV12,DobzinskiV12,DanielySS15} (the latter
assuming RP $\neq$ NP). However, these results no longer apply for mechanisms that are allowed other natural types of queries, e.g., demand queries\footnote{Demand queries are quite natural from an economic point of view as they simply return the most 
valued bundle for the bidder at the given item prices; see Section~\ref{sec:prelim}.}. This has led the researchers to study
the communication complexity of this problem that can capture arbitrary queries to valuations~\cite{Nisan00,BlumrosenN02,DobzinskiNS05,NisanS06,DobzinskiV13,DobzinskiNO14,Dobzinski16b,Assadi17ca,BravermanMW17,EzraFNTW18}. Although a clear 
path for proving a separation between the communication complexity of truthful mechanisms and general algorithms was shown recently in~\cite{Dobzinski16b} (see also~\cite{BravermanMW17,EzraFNTW18}), no such separation is still known.


\section{Preliminaries}\label{sec:prelim}

\paragraph{Notation.} We denote by $N$ the set of bidders and by $M$ the set of items.  We use bold-face letters to denote vectors of prices and capital letters for allocations. For a price vector $\bprice$ and a set of items $M' \subseteq M$, we define $\bprice(M') := \sum_{j \in M'} \price_j$. 
For an allocation $A = (A_1,\ldots,A_n)$, we sometimes abuse the notation and use $A$ to denote the set of allocated items. A restriction of allocation $A$ to bidders in $N' \subseteq N$ and items $M' \subseteq M$ 
is an allocation $A'$ consisting of $A_i \cap M'$ for every $i \in N'$. 

\subsection{Submodular and XOS Valuation Functions}\label{sec:valuations}

We make the standard assumption that valuation $v_i$ of each bidder $i$ is normalized, i.e., $v_i(\emptyset)=0$, and  monotone, i.e., $v_i(S) \leq v_i(T)$ for every $S \subseteq T \subseteq M$. 
We are interested in the case when bidders valuations are \emph{submodular} and hence capture the notion of ``diminishing marginal utility'' of items for bidders.  
A valuation  $v$ is submodular iff  $v(S \cup T) + v(S \cap T) \leq v(S) + v(T)$ for any $S,T \subseteq M$. 

Submodular functions are a strict
subset of \emph{XOS} valuations also known as \emph{fractionally additive} valuations (see, e.g.~\cite{Feige06,LehmannLN06}) defined as follows. 
A valuation $a$ is additive iff $a(S) = \sum_{j \in S} a(\set{j})$ for every bundle $S$. A valuation function $v$ is XOS iff there exists $t$ additive valuations $\set{a_1,\ldots,a_t}$ such that $v(S) = \max_{r \in [t]} a_r(S)$ 
for every $S \subseteq M$. Each $a_r$ is referred to as a \emph{clause} of $v$. If $a \in \argmax_{r \in [t]} a_r(S)$, then $a$ is called a \emph{maximizing clause} for $S$ and $a(\set{j})$ is a \emph{supporting price} of 
item $j$ in this maximizing clause. We say that an allocation $A = (A_1,\ldots,A_n)$ of items to $n$ bidders with XOS valuation is \emph{supported} by prices $\bq = (q_1,\ldots,q_m)$ iff each $q_j$ is a supporting price
for item $j$ in the maximizing clause of the bidder $i$ to whom $j$ is allocated, i.e., $j \in A_i$.

\paragraph{Query access to valuations.} Since valuations have size exponential in $m$, a common assumption is that valuations are specified via certain queries instead, in particular, value queries and demand queries. 
A value query to valuation $v$ on bundle $S$ reveals the value of $v(S)$. A demand query specifies a price vector $\bprice$ on items and the
answer is the ``most demanded'' bundle under this pricing, i.e., a bundle $S \in \argmax_{S'} \{v(S')-\bprice(S')\}$.

\subsection{A Fixed-Price Auction}\label{sec:fixed-price}

We use a standard fixed-price auction as a subroutine in our mechanism. For an \emph{ordered} set $N$ of bidders, $M$ of items, and a price vector $\bprice$, $\spmech(N,M,\bprice)$ is defined as follows. 
  \begin{tbox}
	\underline{$\spmech(N,M,\bprice)$}
	\begin{enumerate}
		\item Iterate over the bidders $i$ of the ordered set $N$ in the given order:
		\begin{enumerate}
			\item Allocate $A_i \in \argmax_{S \subseteq M} \{ v_i(S) - \bprice(S) \}$ to bidder $i$ and update $M \leftarrow M \setminus A_i$. 
		\end{enumerate}
		\item Return the allocation $A = (A_1,\ldots,A_n)$. 
	\end{enumerate}
\end{tbox}

It is easy to see that $\spmech$ can be implemented using one demand query per bidder. Its truthfulness is also easy to check as  bidders have no influence on the pricing mechanism. 

The following lemma gives a key property of this auction used in our proofs. Variants of this lemma have already appeared in the literature, e.g., in~\cite{DobzinskiNS06,Dobzinski07,FeldmanGL15,Dobzinski16,EhsaniHKS18} (although we are not aware of this exact
statement). For completeness, we prove this lemma in Appendix~\ref{app:lem-fixed-price}. 
\begin{lemma}\label{lem:fixed-price}
	Let $A:=\spmech(N,M,\bprice)$ and $\delta <1/2$. 
	Suppose $O$ is an allocation with supporting prices $\bq$ and $\Ms$ is the set of items $j$ with $\delta \cdot q_j \leq p_j < \frac{1}{2} \cdot q_j$. 
	Then, $\val{A} \geq \delta \cdot \bq(\Ms)$. 
\end{lemma}


\section{The High-Level Overview}\label{sec:overview}

We describe our mechanism using three parameters $\alpha := \Theta(1)$, $\beta := O(\log\log{m})$, and $\gamma := \Theta(\alpha \beta)$. 
Let $O=(O_1,\ldots,O_n)$ be an optimal allocation with welfare $\OPT$ and $\bq = (q_1,\ldots,q_m)$ be its supporting prices (obviously, $O$ and $\bq$ are unknown). For now, let us assume that every $q_j$ belongs to $\set{1,\gamma,\gamma^2,\ldots,\gamma^{K}}$, for some $K=O(\log{m})$ (and hence prices
are roughly $\poly{(m)}$ large).

The crux of our mechanism is to ``learn'' $\bq$, namely, find another price vector $\bprice$ such that 
for some subset $C \subseteq M$ with $\bq(C) \approx \val{O}$, $\bprice$ \emph{point-wise} $\gamma$-approximates $\bq$ for items in $C$ (i.e., within a multiplicative factor of $\gamma$). 
Having learned such prices, we can  run a fixed-price auction with prices $\bprice$, and by Lemma~\ref{lem:fixed-price}, obtain an allocation with welfare $\approx \gamma \cdot \val{O}$. 

In order to obtain the price vector $\bprice$, we start with a rough guess $\bpricei{1}$ for what prices should be (say, all ones),  and update
our guess over (at most) $\beta$ \emph{iterations}. In each iteration $i \in [\beta]$, we use the prices $\bpricei{i}$ learned so far to find $\alpha$ new price vectors $\bpricei{i}_1,\ldots,\bpricei{i}_\alpha$, and
``explore'' for \emph{each} item~$j\in M$ which of these $\alpha$ vectors best represents its price in $\bq$, and then assign that price to  item $j$ in $\bpricei{i+1}$. We continue this for $\beta$ iterations until we converge to the desired price vector $\bprice := \bpricei{\beta+1}$, or we 
decide along the way that the prices learned so far are already ``good enough''. There are three main questions to answer here: $(i)$ how to choose which prices to explore in each iteration, $(ii)$ how to explore
a new price for each item, and finally $(iii)$ how to implement all this in a truthful (and computationally-efficient) manner. 
We elaborate on each part below. 

\paragraph{Part $(i)$ -- which prices to explore.} 
This question can be best answered from the perspective of a single item $j \in M$. 
Originally, we set $\pricei{1}_j \in \bpricei{1}$ to be $1$, and so with our assumption that $q_j\in \set{1,\gamma,\ldots,\gamma^{K}}$,  price $\pricei{i}_j$ will  $(\gamma^K)$-approximate $q_j \in \bq$. We want $\pricei{2}_j$ to 
$(\gamma^{K/\alpha})$-approximate $q_j$ in the next iteration. Thus, we simply need to check for every $\ell \in \set{0,\ldots,\alpha-1}$, whether $q_j \geq \gamma^{\ell \cdot K/\alpha}$ or not (using part $(ii)$ below). 
By picking the largest $\ell^{*}$ for which this is true, we can get a $(\gamma^{K/\alpha})$-approximation to $q_j$. 
As such, for each item, there are only $\alpha$ choices of prices that we need to explore next, which allows us to devise price vectors $\bpricei{1}_1,\ldots,\bpricei{1}_\alpha$ accordingly. 
We repeat the same idea for later iterations as well, maintaining that in iteration $i$, price $\pricei{i}_j \in \bpricei{i}$ will $(\gamma^{K/\alpha^{i-1}})$-approximate $q_j$, and use $\alpha$ 
prices as before in $\bpricei{i}_1,\ldots,\bpricei{i}_\alpha$ to update this to a $(\gamma^{K/\alpha^i})$-approximation for the next iteration. 
This way, after $\beta=O(\log\log{m})$ iterations, we obtain $\pricei{\beta+1}_j$ that $\gamma$-approximates $q_j$ as desired. See Figure~\ref{fig:price-trajectory} for an illustration. 

In the above discussion, we talked about an item $j$ as if its price is learned correctly throughout (i.e., $\pricei{i}_j$ is $(\gamma^{K/\alpha^{i-1}})$-approximating $q_j$ for all $i \in [\beta+1]$). 
Our mechanism cannot guarantee this property for every item (but rather for most of them).  
Moreover, we are also not able to decide which items have been correctly priced, so we simply treat all items as being priced correctly in the mechanism and perform
the above process for them. This means that for some items, their price may have been learned incorrectly in some iteration; so we conservatively ignore their contribution  from now on 
in the analysis. A key part of our analysis is  to show that this does not hurt the performance of the mechanism by much, namely, $\bq(C)$ is  still a good approximation to $\val{O}$, where $C$ is the set of items for which we learn their prices correctly.  

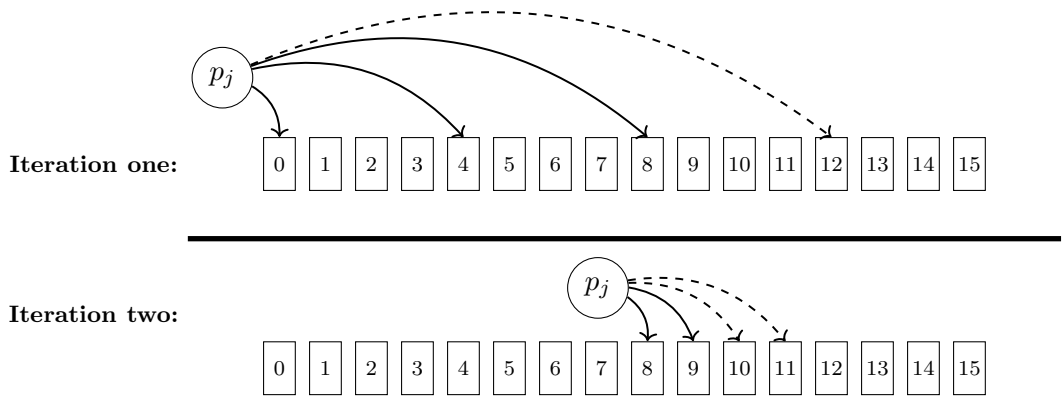
\begin{figure}[h!]
\centering

\begin{tikzpicture}

\node[draw,circle] (1j) {$p_j$};

\node[draw, rectangle, minimum width=12pt, minimum height=20pt, inner sep=0pt] (10) [below right=0.5cm and 0.25cm of 1j]{\scriptsize 0}; 
\foreach \x in {1,...,15}
{
	\pgfmathtruncatemacro{\prev}{\x - 1}
	\pgfmathtruncatemacro{\name}{1\prev}
	\node[draw, rectangle, minimum width=12pt, minimum height=20pt, inner sep=0pt] (1\x) [right=5pt of \name]{\scriptsize \x}; 
}

\node (1it) [left=1cm of 10] {\footnotesize \textbf{Iteration one:}};

\draw[line width=0.75pt, ->, bend left]
(1j) to (10.north);
\draw[line width=0.75pt, ->, bend left]
(1j) to (14.north);
\draw[line width=0.75pt, ->, bend left]
(1j) to (18.north);
\draw[line width=0.75pt, dashed, ->, bend left]
(1j) to (112.north);

\node (1L) [below left=0.5cm and 1cm of 10]{};
\node (1R) [below right=0.5cm and 1cm of 115]{};

\draw[line width=2pt] (1L) -- (1R);


\node[draw,circle] (2j) [below left=1cm and 0.15cm of 18]{$p_j$};

\node[draw, rectangle, minimum width=12pt, minimum height=20pt, inner sep=0pt] (20) [below=2cm of 10]{\scriptsize 0}; 
\foreach \x in {1,...,15}
{
	\pgfmathtruncatemacro{\prev}{\x - 1}
	\pgfmathtruncatemacro{\name}{2\prev}
	\node[draw, rectangle, minimum width=12pt, minimum height=20pt, inner sep=0pt] (2\x) [right=5pt of \name]{\scriptsize \x}; 
}

\node (2it) [below=1.5cm of 1it] {\footnotesize \textbf{Iteration two:}};

\draw[line width=0.75pt, ->, bend left]
(2j) to (28.north);
\draw[line width=0.75pt, ->, bend left]
(2j) to (29.north);
\draw[line width=0.75pt, dashed, ->, bend left]
(2j) to (210.north);
\draw[line width=0.75pt, dashed, ->, bend left]
(2j) to (211.north);

\end{tikzpicture}
\caption{An illustration of the trajectory of the prices of a single item throughout the mechanism. Here, $\alpha=4$ and $\beta=2$. 
Each block $i$ corresponds to price $\gamma^{i}$. Arrows correspond to the price of this item in the corresponding fixed-price auction; a solid arrow means the item was sold, while a dashed arrow means it was not. 
The learned price of this item in this example is $\gamma^{9}$.} \label{fig:price-trajectory}
\end{figure}

\paragraph{Part $(ii)$ -- how to explore a new price.} For this part, we build on a key idea from~\cite{Dobzinski16} in using fixed-price auctions themselves as a ``proxy'' for determining correctness of a guess for  item prices. The idea is as follows: 
suppose we run a fixed-price auction with prices $\bpricei{i}_\ell$ for $\ell \in [\alpha]$ that we want to explore in an iteration $i$. As these prices may be very far from $\bq$ yet, there is no guarantee that this auction returns a high-welfare allocation. However, 
if we choose the ordering of bidders \emph{randomly}, then the \emph{only way} this auction does not succeed in outputting a high-welfare allocation is because it sold \emph{almost all} the items at the current prices (most likely to wrong bidders). 
Hence, an item getting sold in a certain fixed-price auction is a ``good indicator'' that its price in $\bq$ is at least as high as the price used in this fixed-price auction. 
Such an idea was used in~\cite{Dobzinski16} to narrow down the range of item prices from $O(\log{m})$ values to $O(\sqrt{\log{m}})$, which in turn allows the mechanism to simply guess a correct price for each item and achieves an 
$O(\sqrt{\log{m}})$-approximation. 

We take this idea to the next step to obtain our Learnable-Or-Allocatable Lemma (Lemma~\ref{lem:main}). 
Roughly speaking, we show that in each iteration $i$, starting from the set $\Ci{i}$ of correctly priced items, either one of the $\alpha$ auctions for exploring prices will lead to an $O(\beta^2)$-approximate allocation, 
or after this iteration we will manage to further refine the prices of almost all items in $\Ci{i}$. I.e., we obtain a set $\Ci{i+1}$ with $\bq(\Ci{i+1}) \approx \bq(\Ci{i})$ and with $\bpricei{i+1}$ approximating prices $\bq$  for $\Ci{i+1}$ much more accurately than $\bpricei{i}$ (as described in part $(i)$). 
Hence, either during one of the iterations there is an auction that gives us an $O(\beta^2)$-approximation, or we eventually end up with $\bpricei{\beta+1}$ that
point-wise $\gamma$-approximates $\bq$ for a large set of items $\Ci{\beta+1}$. Therefore, by ensuring  $\bq(\Ci{\beta+1}) = \Omega(\OPT)$, a fixed-price auction with prices~$\bpricei{\beta+1}$ gives a $\gamma = O(\alpha\beta)$-approximation by Lemma~\ref{lem:fixed-price}. 

This outline oversimplifies many details. Let us briefly mention two here. Firstly, running fixed-price auctions only help us in not \emph{underpricing}  items for the next iteration; we also need to take care of \emph{overpricing}. 
This is handled by making sure there is a \emph{gap} of $\gamma$ between different prices explored so that not many overpriced items can be sold in an auction.  
Also while for the purpose of this discussion we simply assumed the existence of this gap, in the actual mechanism we need to \emph{create} this gap using a basic randomization idea.  Secondly, our mechanism has no way of determining (in a truthful way) 
which case of the 
Learnable-Or-Allocatable Lemma we are in. This means that  there are $\alpha \cdot \beta$ auctions in the mechanism and any one of them may give an $O(\beta^2)$-approximation welfare. (If not, then we can learn the prices accurately and the final
auction would be an $O(\alpha\beta)$-approximation.) The solution here is then to simply {pick} one of the $(\alpha\beta + 1)$ auctions  \emph{uniformly at random} and allocate according to that.
This way we succeed in finding a good auction with probability at least $1/\alpha \beta$ and hence, in expectation, we obtain an $O(\alpha\beta^3)$-approximation.

\paragraph{Part $(iii)$ -- how to ensure truthfulness.} 
Recall that a fixed-price auction is truthful primarily because the responses of the bidders has no effect on the price of their allocated bundle. However, our mechanism consists of multiple fixed-price auctions and the outcomes of these auctions do influence the 
prices for \emph{later} iterations. As such, to ensure truthfulness, each bidder should only participate in the auctions of a single iteration. Hence, at the beginning of the mechanism, we randomly partition the bidders into $\beta+1$ groups 
$N_1,\ldots,N_{\beta+1}$. Then, in each iteration $i$, we  use the bidders in  group $N_i$ for fixed-price auctions with prices $\bpricei{i}_1,\ldots,\bpricei{i}_\alpha$ to learn prices $\bpricei{i+1}$, and in the final iteration  we run one  fixed-price auction with bidders $N_{\beta+1}$ and prices $\bprice = \bpricei{\beta+1}$.  

This partitioning of bidders results in a key challenge: Our goal in learning the prices should actually be different from what was stated earlier. In particular, the auctions in each iteration $i$ with bidders $N_i$ should reveal the $\bq$ prices of items 
allocated in $O$ to bidders in $N_{>i} := N_{i+1},\ldots,N_{\beta+1}$, \emph{as opposed to} bidders in $N_i$. 
This is because we are no longer able to allocate any item to bidders in $N_1,\ldots,N_{i}$. We handle this also by our Learnable-Or-Allocatable Lemma.
Instead of learning the set $\Ci{i+1}$ with $\bq(\Ci{i+1}) \approx \bq(\Ci{i})$ in the Learnable case, we have a more refined statement in which the LHS is replaced with $\bq$ of  items allocated \emph{only} to bidders in $N_{>i}$. 
This in turn requires a delicate choice of parameters and analysis to balance out  two opposing forces: on one hand, we need $N_i$ to be large enough so that we can ``extrapolate'' the learned prices in auctions with $N_i$ to $N_{>i}$  (in the Learnable case); on the other hand, each $N_i$ should be small enough so that by the time we end up learning the prices, the contribution of the remaining bidders is still large enough.

\paragraph{Comparison to Dobzinski~\cite{Dobzinski16}.} We conclude this section by comparing our work with the previous best $O(\sqrt{\log{m}})$ approximation mechanism of Dobzinski~\cite{Dobzinski16}. 
As stated in Part $(ii)$, our mechanism builds on a key idea from~\cite{Dobzinski16} in using fixed price auctions as a proxy for finding ``good'' prices. On a high level, the main difference between the two works is that Dobzinski~\cite{Dobzinski16} uses fixed price auctions to ``learn'' item prices  in a \emph{single-shot}, but with a relatively poor accuracy. Instead, we use  fixed price auctions \emph{iteratively} in order to learn the  prices of (most) items quite accurately. 

Concretely, assuming that  all prices $q_j\in \set{1,\gamma,\ldots,\gamma^{K}}$ for $K = O(\log m)$, Dobzinski's mechanism can be viewed as a 
special case of our mechanism by setting $\beta=1$ and $\alpha = \sqrt{\log m}$: Dobzinski first uses a set $N_1$ of bidders to run $\alpha = \sqrt{\log m}$ auctions to learn the prices of items (to within an $O(\sqrt{\log{m}})$ factor), 
and then runs one more fixed price auction with these learned prices on bidders $N_{\beta+1}=N_2$. The final  auction is then chosen randomly from these $\sqrt{\log m}+1$ auctions to ensure truthfulness. Considering both the prices
are learned only to within an $O(\sqrt{\log{m}})$ factor and the final auction is chosen from $\sqrt{\log{m}}+1$ auctions, the approximation ratio of this mechanism is $O(\sqrt{\log{m}})$. 
Our mechanism on the other hand learns prices of items in multiple iterations ($\beta=O(\log\log{m})$ iterations) via the Learnable-Or-Allocatable Lemma. This allows us to both use a much smaller number of auctions ($\poly{(\log\log{m})}$ many),
 and at the same time learn prices much more accurately (again to within a $\poly{(\log\log{m})}$ factor), which  ultimately leads to our improved approximation ratio of $O((\log\log{m})^3)$.


\section{The Main Mechanism}\label{sec:mech}

We give our main mechanism for combinatorial auctions with XOS valuations in this section. In the following, we present our  mechanism
with a simplifying assumption (Assumption~\ref{assumption1}). This assumption is made primarily for simplicity of exposition and we show how to remove it in Section~\ref{sec:end-mech}. 

\begin{assumption}\label{assumption1}
	We assume there exists two non-negative numbers $\minprice \leq \maxprice$ such that: 
	\begin{enumerate}[label=(\roman*)]
	\item for every valuation, supporting price of any item for any clause belongs to $\set{0} \cup [\minprice:\maxprice]$;
	\item the ratio of these numbers, denoted by $\ratioprice:= \maxprice/\minprice$, is bounded by some fixed $\poly(m)$. 
	\end{enumerate}
	We further assume that the mechanism is given $\minprice$ and $\maxprice$ as input. 
\end{assumption}

In the following, we first present a simple tree-structure, named the \emph{price tree}, 
that we use in our mechanism for discretizing prices at different scales.  We then describe the method with which we assign different
bidders to different auctions run by our mechanism. Finally, we present our mechanism and prove its computational efficiency and universal truthfulness guarantees. The analysis of the approximation ratio of our mechanism---the main technical contribution of the paper---appears
in the subsequent section. 

\paragraph{Parameters:} 
We define and use the following parameters in our mechanism.  
\begin{itemize}
	\item $\alpha:= \Theta(1)$ -- number of different auctions run in each iteration of our mechanism;
	\item $\beta := O(\log\log{\ratioprice})$ -- number of iterations in our mechanism;
	\item $\gamma := \Theta(\alpha\beta)$ -- the accuracy to which we aim to learn the true prices.
\end{itemize}
Moreover, the above parameters satisfy the following equations:
\begin{align}
&\alpha^{\beta+1} \geq \log_{\gamma}{\ratioprice} \notag \\  &20\alpha\beta \leq \gamma \leq 30\alpha\beta. \label{eq:equations} 
\end{align}
It is immediate to verify that one can choose $\alpha,\beta,\gamma$ satisfying all the above equations. 
\subsection{Price Trees and Their Properties}\label{sec:price-tree}

We define a simple tree-structure used for discretizing the range of prices in $[\minprice : \maxprice]$ by our mechanism. 
The first part is a geometric partition of set of available prices as follows. 

\begin{definition}\label{def:bins}
We  partition $[\minprice : \maxprice]$ into $t:= \ceil{\log_{\gamma}{\ratioprice}}$ \textbf{\emph{bins}} $B_1,\ldots,B_t$ 
where values inside each $B_i$ are within a factor $\gamma$ of each other. We use $\binprice(B_i)$ to denote the min value in $B_i$. 	
\end{definition}

We now use the concepts of bins to define a multi-level partitioning of $[\minprice:\maxprice]$ with different scales of accuracy. 

\begin{definition}[Price Tree]\label{def:price-tree}
	A \textbf{\emph{price tree}} $\TT$ is a rooted tree in which each node $z$ is assigned two attributes: $(i)$ $\bins(z)$ which is a subset of bins $B_1,\ldots,B_t$ with \emph{consecutive} indices, and $(ii)$ $\binprice(z)$ which is the value of $\binprice(B_i)$ 
	where $B_i$ is the \emph{smallest indexed bin} inside $\bins(z)$. The tree $\TT$ satisfies the following properties: 
	\begin{itemize}	 
		\item For the root $z_r$ of $\TT$, $\bins(z_r) := (B_1,\ldots,B_t)$.  
		\item $\TT$ has $t$ leaf-nodes where the $i$-th left most leaf-node $z_i$ of $\TT$ has $\bins(z_i) = B_i$.  
		\item Every non-leaf node $z$ of $\TT$ has $\alpha$ children $z_1,\ldots,z_\alpha$ such that $\bins(z_1)$ contains the first $\alpha$ fractions of $\bins(z)$, $\bins(z_2)$ contains the second $\alpha$ fraction, and so on. 
	\end{itemize}
	By the choice of $\alpha^{\beta+1} \geq t$, the number of levels in $\TT$ is $\beta+1$ (see Figure~\ref{fig:price-tree} for an illustration).
\end{definition}

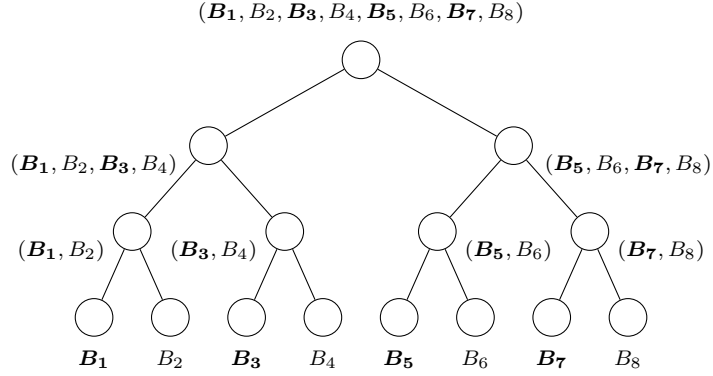
\begin{figure}[t]
\centering
\begin{tikzpicture}
	\node[circle, minimum width=0.5cm, inner sep = 0pt, black, draw](l1){};
	\node[below=0.05cm of l1]{\scriptsize$\bm{B_1}$};
	
	\node[circle, minimum width=0.5cm, inner sep = 0pt, black, draw](l2)[right=0.5cm of l1]{};
	\node[below=0.05cm of l2]{\scriptsize$B_2$};
	
	\node[circle, minimum width=0.5cm, inner sep = 0pt, black, draw](l3)[right=0.5cm of l2]{};
	\node[below=0.05cm of l3]{\scriptsize$\bm{B_3}$};

	\node[circle, minimum width=0.5cm, inner sep = 0pt, black, draw](l4)[right=0.5cm of l3]{};
	\node[below=0.05cm of l4]{\scriptsize$B_4$};

	\node[circle, minimum width=0.5cm, inner sep = 0pt, black, draw](l5)[right=0.5cm of l4]{};
	\node[below=0.05cm of l5]{\scriptsize$\bm{B_5}$};

	\node[circle, minimum width=0.5cm, inner sep = 0pt, black, draw](l6)[right=0.5cm of l5]{};
	\node[below=0.05cm of l6]{\scriptsize$B_6$};

	\node[circle, minimum width=0.5cm, inner sep = 0pt, black, draw](l7)[right=0.5cm of l6]{};
	\node[below=0.05cm of l7]{\scriptsize$\bm{B_7}$};

	\node[circle, minimum width=0.5cm, inner sep = 0pt, black, draw](l8)[right=0.5cm of l7]{};
	\node[below=0.05cm of l8]{\scriptsize$B_8$};

	\draw[opacity=0] (l1) --(l2) node(l12)[midway]{};
	\draw[opacity=0] (l3) --(l4) node(l34)[midway]{};
	\draw[opacity=0] (l5) --(l6) node(l56)[midway]{};
	\draw[opacity=0] (l7) --(l8) node(l78)[midway]{};

	\node[circle, minimum width=0.5cm, inner sep = 0pt, black, draw](a1)[above =0.75cm of l12]{};
	\node[below left=-0.2cm and 0.05cm of a1]{\scriptsize$(\bm{B_1},B_2)$};
	
	\node[circle, minimum width=0.5cm, inner sep = 0pt, black, draw](a2)[above = 0.75cm of l34]{};
	\node[below left=-0.2cm and 0.05cm of a2]{\scriptsize$(\bm{B_3},B_4)$};
	\node[circle, minimum width=0.5cm, inner sep = 0pt, black, draw](a3)[above = 0.75cm of l56]{};
	\node[below right=-0.2cm and 0.05cm of a3]{\scriptsize$(\bm{B_5},B_6)$};
	\node[circle, minimum width=0.5cm, inner sep = 0pt, black, draw](a4)[above = 0.75cm of l78]{};
	\node[below right=-0.2cm and 0.05cm of a4]{\scriptsize$(\bm{B_7},B_8)$};

	\draw[opacity=0] (a1) --(a2) node(a12)[midway]{};
	\draw[opacity=0] (a3) --(a4) node(a34)[midway]{};
	
	\node[circle, minimum width=0.5cm, inner sep = 0pt, black, draw](b1)[above=0.75cm of a12]{};
	\node[below left=-0.2cm and 0.1cm of b1]{\scriptsize $(\bm{B_1},B_2,\bm{B_3},B_4)$};

	\node[circle, minimum width=0.5cm, inner sep = 0pt, black, draw](b2)[above=0.75cm of a34]{};
	\node[below right=-0.2cm and 0.1cm of b2]{\scriptsize$(\bm{B_5},B_6,\bm{B_7},B_8)$};

	\draw[opacity=0] (b1) --(b2) node(b12)[midway]{};

	\node[circle, minimum width=0.5cm, inner sep = 0pt, black, draw](c1)[above=0.75cm of b12]{};
	\node[above = 0.1cm of c1]{\scriptsize$(\bm{B_1},B_2,\bm{B_3},B_4,\bm{B_5},B_6,\bm{B_7},B_8)$};
	
	\draw
	(l1) -- (a1)
	(l2) -- (a1)
	(l3) -- (a2)
	(l4) -- (a2)
	(l5) -- (a3)
	(l6) -- (a3)
	(l7) -- (a4)
	(l8) -- (a4)
	(a1) -- (b1)
	(a2) -- (b1)
	(a3) -- (b2)
	(a4) -- (b2)
	(b1) -- (c1)
	(b2) -- (c1);

\end{tikzpicture}
\caption{An illustration of a price tree $\TT$ with $\alpha=2$, $\beta=2$, and $t=8$. By considering only the bold-face bins (with odd indices), we obtain the modified price tree $\TT^o$.} \label{fig:price-tree}
\end{figure}

A price tree $\TT$ gives a nested partitioning of the range $[\minprice:\maxprice]$ into $\beta+1$ levels with different granularities. 
We say that a price $p$ \emph{belongs} to a node $z$ of $\TT$ iff $p$ appears in one of the bins in $\bins(z)$; moreover, if $p = \binprice(z)$, 
then we say $p$ \emph{strongly belongs} to $z$. 

\begin{definition}\label{def:level-price}
We say a price vector $\bprice=(p_1,\ldots,p_m)$ is a \textbf{\emph{level-$i$ price vector}}  iff every $p_j$ strongly belongs to some node $z_j$ in level $i$ of $\TT$.  

We assign $\alpha$ \textbf{\emph{canonical level-$(i+1)$ price vectors}} $\bprice_1,\ldots,\bprice_\alpha$ to a level-$i$ price vector $\bprice$, where in $\bprice_k=(p'_{1},\ldots,p'_{m})$ 
each $p'_j$ strongly belongs to the $k$-th child of $z_j$ to which $p_j$ strongly belongs.  
\end{definition}

\paragraph{Modified price trees.} Using price trees in our mechanism directly is problematic primarily because 
 it is possible that a price $p \in B_i$ for some bin $B_i$ is actually closer to $\binprice{(B_{i+1}})$ than $\binprice{(B_i)}$, hence making the learning of the correct bin for $p$ not feasible. 

To fix this issue, we consider the following two modified price trees $\TTodd$ and $\TTeven$ instead: 
$\TTodd$ is a subtree of $\TT$ obtained by retaining only the \emph{odd} indexed bins $B_1,B_3,\ldots$ in $\bins$ of $\TT$; $\TTeven$ is defined analogously by retaining all \emph{even} indexed bins. 
In our mechanism, we pick one of $\TTodd$ or $\TTeven$ at random and from there on, only consider the prices that belong to the bins that appear in the corresponding modified price tree. This way, for any two price $p,p'$ that belong
to two different nodes of the modified tree, there is at least a factor $\gamma$ gap between $p$ and $p'$. 

\subsection{Partitioning Bidders}\label{sec:bidder-partition}

Our main mechanism involves partitioning the set of bidders into $\beta+1$ different \emph{groups} $N_1,\ldots,N_{\beta+1}$ and assigning them to different auctions throughout the mechanism:

\begin{itemize}
	\item $\partition(N)$: Let $N' \leftarrow N$ and for $i=1$ to $\beta$ iterations: pick a random permutation of $N'$ and insert the first $|N'|/(10\beta)$ bidders into  $N_i$; update $N' \leftarrow N' \setminus N_i$. At the end, 
	return $N_1,\ldots,N_\beta$ and $N_{\beta+1} := N'$. 
\end{itemize}

\noindent
We note that size of $N_1,\ldots,N_{\beta}$ are \emph{decreasing} in expectation, while $N_{\beta+1}$ is larger than the rest. 

\subsection{Formal Description of the Mechanism}\label{sec:main-mech}

We are now ready to give our main mechanism under Assumption~\ref{assumption1}. For that, we also need the following procedure first: 

\begin{itemize}
	\item $\pupdate(\Ai{i}_1,\ldots,\Ai{i}_{\alpha},\bpricei{i}_1,\ldots,\bpricei{i}_\alpha)$: For any item $j \in M$, we let $p'_j$ be equal to $p_j \in \bpricei{i}_k$ where $k$ is the \emph{largest} index such that item $j$ is allocated in $\Ai{i}_k$ (if $j$ is never allocated, we set $k=1$). Return $\bpricei{i+1} = (p'_1,\ldots,p'_m)$.  
\end{itemize}

\noindent
We now define our mechanism.

\begin{tbox}
\underline{$\mech{(N,M)}$} 
\begin{enumerate}
	\item Let $(N_1,N_2,\ldots,N_{\beta+1}) := \partition(N)$.
	\item Pick one of the modified trees $\TTodd$ or $\TTeven$ uniformly at random and denote it by $\TTstar$.   
	\item Let $\bpricei{1}$ be the (unique) level-$1$ (root) price of $\TTstar$. For $i = 1$ to $\beta$ \underline{iterations}: 
	\begin{enumerate}
		\item Let $\bpricei{i}_1,\ldots,\bpricei{i}_{\alpha}$ be the level-$(i+1)$ canonical price vectors of $\bpricei{i}$ in $\TTstar$ (Definition~\ref{def:level-price}). 
		\item For $j=1$ to $\alpha$: run $\spmech(N_i,M,\frac{\bpricei{i}_j}{2})$ and let $\Ai{i}_j$ be the allocation. 
		\item\label{line:coin-toss} W.p. $(1/\beta)$, pick $\jstar \in [\alpha]$ uniformly at random and return $\Ai{i}_{\jstar}$ as the final allocation; 
		otherwise, let $\bpricei{i+1} := \pupdate(\Ai{i}_1,\ldots,\Ai{i}_{\alpha},\bpricei{i}_1,\ldots,\bpricei{i}_\alpha)$, and continue. 
	\end{enumerate}
	\item Run $\spmech(N_{\beta+1},M,\frac{\bpricei{\beta+1}}{2})$ and return the allocation $A^*$. 
\end{enumerate}
\end{tbox}

We shall right away remark that in $\mech$, every price vector $\bpricei{i}$ computed in iteration $i$ is a level-$i$ price vector and hence the canonical price vectors defined in each iteration indeed do exist. We have the following
theorem which is the main technical result of this paper.

\begin{theorem}\label{thm:main-mech}
	For a combinatorial auction with $n$ submodular (even XOS) bidders and $m$ items, under Assumption~\ref{assumption1}, $\mech$ is universally truthful, uses $O(n)$ demand queries and polynomial time, and achieves an approximation ratio of $O((\log\log{m})^3)$ in expectation. 
\end{theorem}
We  remark that  our mechanism in Theorem~\ref{thm:main-mech} only makes $O(1)$ queries to the valuation of each bidder, which is clearly optimal, and results in a {highly efficient} mechanism (computationally).

To see that $\mech$ is truthful, notice that every bidder $b$ is participating in at most $\alpha$ fixed-price auctions of $\spmech$ for which the prices of items have already been fixed entirely independent of $b$'s valuations (and responses). 
Moreover, for bidders in $N_1,\ldots,N_{\beta}$ that participate in more than one auction, the choice of which items (if any) they are being allocated across the auctions is entirely independent of the auction outcome and is determined by the random
coin tosses in Line~(\ref{line:coin-toss}). This still does not imply that truth telling is a dominant strategy as a bidder can ``threat'' another bidder by presenting wrong valuations in subsequent auctions they both participate in (see, e.g.~\cite{Dobzinski07,Dobzinski16}). 
To fix this, we make each bidder $b$ output the preferences in all fixed-price auctions $b$ participates in \emph{simultaneously} (or alternatively hide bidders responses from each other). As was observed in~\cite{Dobzinski07,Dobzinski16} this ensures the truthfulness of the mechanism.  

Computational efficiency of $\mech$ and the bound on number of demand queries follow immediately from the fact that each bidder is participating in at most $\alpha = \Theta(1)$ fixed-price auctions, each of which requires one demand query per bidder.


\newcommand{\Ngi}{\ensuremath{N_{\geq i}}}
\newcommand{\Nggi}{\ensuremath{N_{> i}}}
\newcommand{\ODk}{\ensuremath{O^{D}_{\geq k}}}
\newcommand{\ODkb}{\ensuremath{O^{D}_{\geq k,b}}}
\newcommand{\Alk}{\ensuremath{A_{< k}}}
\newcommand{\ngi}{n_{\geq i}}

\section{The Analysis of Main Mechanism} \label{sec:analysis}

We now present the analysis of the approximation ratio of $\mech$. 

\paragraph{Notation.} To avoid confusion, throughout this section, we use ``$i$'' to index the iterations, ``$j$'' to index the auctions inside an iteration, ``$b$'' to index the bidders, and ``$\ell$'' to index the items.  

We pick an optimal allocation $O = (O_1,\ldots,O_n)$ of items with supporting prices $\bq = (q_1,\ldots,q_m)$ and denote by $\OPT$ the welfare of this allocation. 
We further define $\Ostar$ as the restriction of $O$ to items with supporting prices in $\bq$ that belong to the modified price tree $\TTstar$ chosen by the mechanism. 
Similarly, $\bqstar$ is defined by zeroing out the price of items in $\bq$ that are not allocated by $\Ostar$ and leaving the rest unchanged. 
We also define the following series of refinement of $\bq$ based on the bidders in $N_1,\ldots,N_{\beta+1}$ and the choice of $\TTstar$. For every $i \in [\beta+1]$, 
$\bqi{i} = (\qi{i}_1,\ldots,\qi{i}_m)$ is defined so that for every item $\ell \in M$, $\qi{i}_\ell = 0$ iff $\ell$ is allocated in $\Ostar$ to some bidder in $N_{1},\ldots,N_{i-1}$ or is not allocated at all, 
and otherwise $\qi{i}_\ell = q_\ell$ for $q_\ell \in \bqstar$.  

Fix any iteration $i$ and the price vector $\bpricei{i} = (\pricei{i}_1,\ldots,\pricei{i}_m)$ obtained by the mechanism so far. We say an item $\ell \in M$ is \emph{correctly priced} in iteration $i$ iff 
$\pricei{i}_\ell$ belongs to the same level-$i$ node in $\TTstar$ as $\qi{i}_\ell$. Note that by construction, $\pricei{i}_\ell$ always strongly belongs to a node,  and hence for any correctly priced item, we have 
$\pricei{i}_\ell \leq \qi{i}_\ell$. We use $\Ci{i}$ to denote the set of all items that are correctly priced throughout \emph{all} iterations $1$ to $i$. Hence, under this definition, $\Ostar=\Ci{1} \supseteq \Ci{2} \supseteq \ldots \supseteq \Ci{\beta+1}$. 
The definition of the price tree ensures that by moving from $\Ci{1}$ towards $\Ci{\beta+1}$ we are learning the prices of correctly prices items more and more accurately. 

\paragraph{Learnable-Or-Allocatable Lemma.} The goal of our mechanism is to learn a set $\Ci{\beta+1}$ such that $\bqi{\beta+1}(\Ci{\beta+1})$ is still sufficiently large compared to $\bqstar(\Ostar)$. 
Having reached such a state, we can run a fixed price auction 
with price vector $\bpricei{\beta+1}/2$ with bidders $N_{\beta+1}$. Since for items in $\Ci{\beta+1}$, their price in $\bpricei{\beta+1}$ and $\bqi{\beta+1}$ are within a $\gamma$ factor of each other,
we can invoke Lemma~\ref{lem:fixed-price} and obtain an allocation with welfare at least $\gamma$ fraction of  $\bqi{\beta+1}(\Ci{\beta+1})$. 

Of course, in general, it is too much to expect that our mechanism can converge to a particular price vector $\bqstar$ (think of a case where 
there are many different optimal allocations with different prices; converging to one such price vector necessarily means not converging to the other ones). 
The following lemma, which is the heart of the proof, however states that in each iteration, we can either ``learn'' the prices of most items more accurately than before, or we can already ``allocate'' the items 
efficiently enough at the current prices.

\begin{Lemma}[\textbf{Learnable-Or-Allocatable Lemma}]\label{lem:main}
	For any iteration $i \in [\beta]$, conditioned on any outcome of first $i-1$ iterations and choice of $\TTstar$: 
\begin{enumerate}[label=(\roman*)]
	\item\label{item:main-learnable} either $\expect{\bqi{i+1}(\Ci{i+1})} \geq \bqi{i}(\Ci{i}) - \frac{\opt}{3\beta}$, where the expectation is over $N_i$; 
	\item\label{item:main-allocatable} or $\expect{\val{\Ai{i}_{\jstar}}} \geq \frac{\OPT}{200\alpha \cdot \beta^2}$, where the expectation is over  $N_i$ and $\jstar \in [\alpha]$.
\end{enumerate} 
We refer to the first case as \textbf{Learnable} and to the second one as \textbf{Allocatable}. 
\end{Lemma}

We prove Lemma~\ref{lem:main} next and then use it to conclude the proof of Theorem~\ref{thm:main-mech}. 

\subsection{Proof of Lemma~\ref{lem:main}  -- Learnable-Or-Allocatable Lemma}\label{sec:main-proof}

We start with a high level overview. We prove this lemma in three steps: 
\begin{enumerate}[label=($\roman*$)]
	\item \emph{No underestimating prices:} We first show (Lemma~\ref{lem:random-arrival}) that for any of the auctions in this iteration, either most of the correctly priced items  (with respect to this auction) 
	are sold, or this auction itself can result in a high welfare. This step allows us to argue that for many of the items we can sell them in these auctions with a price \emph{at least as high} as their true price, and hence 
	we will not underestimate their prices in this iteration. The proof of this part crucially uses the fact that the bidders are coming in a random order and is along the  lines of a similar argument by Dobzinski~\cite{Dobzinski16}. 
	\item \emph{No overestimating prices:} We then show (Lemma~\ref{lem:no-overestimate}) that in these auctions only a small fraction of items may continue to get sold even past their correct price. Roughly speaking, this is because if 
	we could actually sell many items in auctions with higher prices, this implies that the true welfare of the auction is larger than $\OPT$, a contradiction. This part  relies on 
	the ``price gap'' we introduced in price trees by picking $\TTeven$ or $\TTodd$ (instead of $\TT$ itself). 
	\item \emph{Handling removed bidders:} Finally, in Claim~\ref{clm:qi-qi+1} we argue that even if  we ignore the items for bidders in $N_i$ (as the mechanism no longer considers these bidders), the \emph{remaining} correctly priced items 
	still have a substantial contribution. This part of the proof uses the fact that we only consider a small random subset $N_i$ of the remaining bidders. 
\end{enumerate}
 
We now present the formal proof.  
Throughout the proof, we fix $i \in [\beta]$ and condition on the outcome of the first $i-1$ iterations and the choice of $\TTstar$. 
Conditioning on the outcome of the first $i-1$ iterations fixes the set of bidders $N_1,\ldots,N_{i-1}$ but bidders in $N_i$ are chosen randomly from the remaining bidders. 
Fixing the bidders $N_1,\ldots,N_{i-1}$ also fixes the price vector $\bqi{i}$. This conditioning also fixes the level-$i$ price vector $\bpricei{i}$ and its canonical level-$(i+1)$ price vectors 
$\bpricei{i}_1,\ldots,\bpricei{i}_{\alpha}$. The set $\Ci{i}$ of items that have been  correctly priced so far is also fixed. 

We partition the correctly priced items $\Ci{i}$ into $\alpha$ sets $\Di{i}_1,\ldots,\Di{i}_\alpha$, defined as follows. For an item $\ell \in \Ci{i}$, let $z_\ell$ denote the node in level $i$ of $\TT$ that both $\pricei{i}_\ell$ and $\qi{i}_\ell$ 
belong to. Suppose the child-node of $z_\ell$ to which $\qi{i}_\ell$ belongs is $z_{\ell,j}$ for some $j \in [\alpha]$. We place  item $\ell$ in $\Di{i}_j$ in this case. Note that under this partitioning, the level $(i+1)$ node
$z_{\ell,j}$ to which $\qi{i}_\ell$ belongs is the same node that $p_\ell \in \bpricei{i}_{j}$ (strongly) belongs to; thus, for items in $\Di{i}_j$, $\bpricei{i}_j \leq \bqi{i}$.  

In the following lemma, we use the construction of $\partition$ to argue that for any $j \in [\alpha]$, we either allocate most
items in $\Di{i}_j$ in the fixed price auction with price vector $\bpricei{i}_j$ or otherwise this auction is obtaining a large welfare. 

\begin{lemma}\label{lem:random-arrival}
	For any $j \in [\alpha]$, we have 
	$20\beta \cdot \expect{\val{\Ai{i}_j}} + \expect{\bqi{i}(\Ai{i}_j \cap \Di{i}_j)} \geq \bqi{i}(\Di{i}_j). $

\end{lemma}
\begin{proof}
	We define $\Ngi:= N \setminus (N_1 \cup \ldots \cup N_{i-1})$. In the following, all expectations  are taken over the choice of $N_i$ from $\Ngi$. Recall that in $\partition$, $N_i$ is chosen from $\Ngi$ 
	by picking a random permutation and picking the first $\card{\Ngi}/(10\beta)$ bidders in $N_i$. 
	
	Define $\Os_{N_i}$ as the restriction of $\Ostar$ to items in $\Di{i}_j$ and bidders in $N_i$. Similarly, define $\Os_{\Nggi}$ as the restriction of $O$ to items in $\Di{i}_j$ and bidders in $\Nggi:= \Ngi \setminus N_i$. 
	Note that $\bqi{i}(D_j) = \bqi{i}(\Os_{N_i}) + \bqi{i}(\Os_{\Nggi})$ (recall that $\bqi{i}$ gives price $0$ to items not allocated to bidders in $\Ngi$). 
	
	Proof of this lemma is by a simple combination of the following two claims. 
	\begin{claim}\label{clm:random-arrival-1}
		Deterministically, $\val{\Ai{i}_j} \geq \bqi{i}(\Os_{N_i} \setminus \Ai{i}_j)/2$. 
	\end{claim}
	\begin{proof}
		For any bidder $b \in N_i$, when it was bidder $b$'s turn to pick a set in allocation $\Ai{i}_j$ of $\spmech(N_i,M,\bpricei{i}_j/2)$, $b$ could have picked $\Os_b \setminus \Ai{i}_j \subseteq \Os_{N_i}$ and obtain the 
		profit of 
		\begin{align*}
		v_b(\Os_b \setminus \Ai{i}_j) - \bpricei{i}_j(\Os_b \setminus \Ai{i}_j)/2 ~~ \geq ~~ \bqi{i}(\Os_b \setminus \Ai{i}_j) - \bpricei{i}_j(\Os_b \setminus \Ai{i}_j)/2 ~~ \geq ~~ \bqi{i}(\Os_b \setminus \Ai{i}_j)/2.
		\end{align*}
		The first inequality is because $\bqi{i}$ is a supporting price for $\Os_b \setminus \Ai{i}_j$ and the second one is because $\bpricei{i}_j \leq \bqi{i}$ on the items in $\Di{i}_j$. 
		As bidder $b$ maximizes the profit by picking $\Ai{i}_{j,b}$, we have
		\begin{align*}
			\val{\Ai{i}_j} \quad = \quad \sum_{b \in N_i} v_b(\Ai{i}_{j,b}) \quad \geq \quad  \sum_{b \in N_i} \bqi{i}(\Os_b \setminus \Ai{i}_j)/2 \quad = \quad \bqi{i}(\Os_{N_i} \setminus \Ai{i}_j)/2.  \Qed{Claim~\ref{clm:random-arrival-1}}
		\end{align*} 
		
	\end{proof}
	\begin{claim}\label{clm:random-arrival-2}
		By randomness of choice of $N_i$ from $\Ngi$, $\expect{\val{\Ai{i}_j}} \geq (\frac{1}{10\beta}) \cdot \expect{\bqi{i}(\Os_{\Nggi} \setminus \Ai{i}_j)/2}$. 
	\end{claim}
	\begin{proof}
	For the purpose of this proof, it helps to think of picking $N_i$ alternatively by repeating the following for $n_i:=\card{N_i}$ \emph{steps}:
	 sample a bidder uniformly at random from $\Ngi$, include it in $N_i$ , and remove it from consideration for sampling from now on. It is immediate that the distribution of $N_i$ is the same under this and the original definition.  
	 
	For every $k \in [n_i]$, define $N_{i,k} \subseteq N_i$ as the set $N_i$ constructed \emph{before} step $k$ and $\ODk$ as the restriction of $\Ostar$ to $\Di{i}_j$ and $\Ngi \setminus N_{i,k}$. Thus, $\ODk \supseteq \Os_{\Nggi}$ and 
	 hence $\bqi{i}(\ODk) \geq \bqi{i}(\Os_{\Nggi})$ for every $k$. Recall that $\spmech$ operates in a greedy manner and hence allocation of bidders participating  in the auction before step $k$ are already determined by
	 step $k$. Define $\Alk$ as the set of items allocated by auction \emph{before} step $k$ and let $u_k := v_b(\Ai{i}_{j,b})$ where $b$ is the chosen bidder in step $k$ 
	 and $\Ai{i}_{j,b}$ is the allocation $b$ will get by participating in $\spmech(N_i,M,\bpricei{i}_j/2)$. 
	 
	We first prove that $u_k \geq \bq(\ODkb \setminus \Alk)/2$. This is precisely because of the same reason as in Claim~\ref{clm:random-arrival-1} that $b$ could have chosen $\ODkb \setminus \Alk$ but decided to pick another set. 
	Define $\ngi := \card{\Ngi}$. Recall that $b$ is chosen uniformly at random from the $(\ngi-k+1)$ bidders at step $k$ and hence, 
	\begin{align}
	\Exp_{b}\bracket{u_k} \quad \geq \quad \frac{1}{\ngi-k+1} \cdot \frac{\bqi{i}(\ODk \setminus \Alk)}{2}  \quad \geq \quad \frac{1}{\ngi} \cdot \frac{\bqi{i}(\ODk \setminus \Alk)}{2}. \label{eq:rand1}
	\end{align}
	We can thus write, 
	\begin{align*}
	\Exp_{N_i}\bracket{\val{\Ai{i}_j}} &= \sum_{k=1}^{n_i}   \Exp_{N_{i,k}}\Exp_{b}\bracket{u_k \mid N_{i,k}}  \\
	&\geq  \frac{1}{2\ngi} \cdot \sum_{k=1}^{n_i} \cdot \Exp_{N_{i,k}}\bracket{\bqi{i}(\ODk \setminus \Alk) \mid N_{i,k}} \tag{by Eq~(\ref{eq:rand1})}\\
	&\geq \frac{1}{2\ngi} \cdot \sum_{k=1}^{n_i}  \Exp_{N_{i}}\bracket{\bqi{i}(\Os_{\Ngi} \setminus \Ai{i}_j)} \tag{as $\Os_{\Ngi} \subseteq \ODk$ and $\Ai{i}_j \supseteq \Alk$ always} \\
		&= \frac{n_i}{\ngi} \cdot   \Exp_{N_{i}}\bracket{\bqi{i}(\Os_{\Ngi} \setminus \Ai{i}_j)/2} 
	\quad  = \quad  \left(\frac{1}{10\beta}\right) \cdot \expect{\bqi{i}(\Os_{\Ngi} \setminus \Ai{i}_j)/2}. \Qed{Claim~\ref{clm:random-arrival-2}}
	\end{align*}
		
	\end{proof}
	
	We can now conclude the proof of Lemma~\ref{lem:random-arrival} as follows. By Claims~\ref{clm:random-arrival-1} and~\ref{clm:random-arrival-2}, 
	\begin{align*}
		20\beta \cdot \expect{\val{\Ai{i}_j}} &\geq \expect{\bqi{i}(\Os_{N_i} \setminus \Ai{i}_j) + \bqi{i}(\Os_{\Ngi} \setminus \Ai{i}_j)} \\
		&= \expect{\bqi{i}(\Di{i}_j \setminus A_j)} \quad = \quad \bqi{i}(\Di{i}_j) - \expect{\bqi{i}(\Di{i}_j \cap \Ai{i}_j)}.
	\end{align*}
	This concludes the proof. \Qed{Lemma~\ref{lem:random-arrival}}
	
\end{proof}

The quantity $\Ai{i}_j \cap \Di{i}_j$ bounded in Lemma~\ref{lem:random-arrival} is closely related to the set of correctly priced items at iteration $i+1$, namely $\Ci{i+1}$. The only 
difference between the two sets is that some items in $\Ai{i}_j \cap \Di{i}_j$ can be allocated even in $\Ai{i}_k$ for $k > j$ and hence in $\pupdate$, we assign a larger price to them. 
In the following, we prove that the contribution of such items cannot be too large.

\begin{lemma}\label{lem:no-overestimate}
	We have $\bqi{i}(\Ci{i+1}) \geq \sum_{j=1}^{\alpha}\bqi{i}(\Ai{i}_j \cap \Di{i}_j) - \frac{\OPT}{10\beta}.$
\end{lemma}
\begin{proof}

By definition of $\pupdate$, we know that items in $\Ai{i}_j \cap \Di{i}_j$ will join $\Ci{i+1}_j$ iff they do not belong to some $\Ai{i}_k$ for $k > j$. 
For each $k > j$, let $\OE_k$ be the set of items in $\Di{i}_j \cap \Ai{i}_j$ that are also allocated in $\Ai{i}_k$. Then, $\OE_{j+1} \cup \ldots \cup\OE_{\alpha}$ forms the set of all items in $\Di{i}_j$ that the mechanism \emph{overestimates} their price in iteration $i$.  We bound the contribution of such items. 

Fix some $k > j$. Consider the level $i+1$ of the price tree $\TTstar$. There are $\alpha^i$ nodes in this level to which the $\bqi{i}$-price  of an item in $\Di{i}_j$ can belong to. Let $oe_{k,1},\ldots,oe_{k,\alpha^i}$ be the number of items corresponding to these 
nodes that were allocated in $\Ai{i}_k$ as well. Hence, $\card{\OE_{k}} = \sum_{\ell=1}^{\alpha^i} oe_{k,\ell}$. Moreover, let $p_{k,1},\ldots,p_{k,\alpha^i}$ be the \emph{maximum} prices that belong to these nodes. 
Finally, let $p'_{k,1},\ldots,p'_{k,\alpha^i}$ be the prices that these items were sold in $\Ai{i}_k$. Because $\TTstar$ is either $\TTodd$ or $\TTeven$, we have $p_{k,\ell} \leq \gamma^{k-j} \cdot p'_{k,\ell}$ (there is a factor $\gamma$ gap between the maximum 
price of any bin $B_x$ and minimum price of $B_{x+2}$).  

Since all the items in $\OE_{k}$ are sold in a single application of $\spmech$, we know that there exists an allocation with supporting prices $p'_{k,\ell}$ for $oe_{k,\ell}$ items for all $\ell \in [\alpha^i]$. 
As such, 
\begin{align}
	\OPT \quad \geq \quad \sum_{\ell=1}^{\alpha^i} p'_{k,\ell} \cdot oe_{k,\ell} \quad \geq \quad \gamma^{k-j} \cdot \sum_{\ell=1}^{\alpha^i} p_{k,\ell} \cdot oe_{k,\ell} \quad \geq \quad  \gamma^{k-j} \cdot \bqi{i}(\OE_k),  \label{eq:no-overestimate}
\end{align}
by definition of $p_{k,\ell}$ as the maximum price inside the nodes of $\TTstar$ that prices of $\bqi{i}(\OE_k)$ belong to. Summing up Eq~(\ref{eq:no-overestimate}) for all choices of $k > j$, we have 
\begin{align*}
	\sum_{k=j+1}^{\alpha} \bqi{i}(\OE_k) \quad \leq \quad \sum_{k=j+1}^{\alpha}\frac{1}{\gamma^{k-j}} \cdot \OPT \quad \leq \quad \frac{2}{\gamma} \cdot \OPT. 
\end{align*}
Finally, as there are $\alpha$ choices for $j$, we have
\begin{align*}
	\paren{\sum_{j=1}^{\alpha}\bqi{i}(\Ai{i}_j \cap \Di{i}_j)} - \bqi{i}(\Ci{i+1}) \quad \leq \quad \sum_{j=1}^{\alpha} \frac{2}{\gamma} \cdot \OPT \quad = \quad \frac{2\alpha}{\gamma} \cdot \OPT \quad \leq \quad \frac{\OPT}{10\beta},
\end{align*}
by the choice of $\gamma \geq 20\alpha\beta$ in Eq~(\ref{eq:equations}). \Qed{Lemma~\ref{lem:no-overestimate}}

\end{proof}
So far we only considered prices with respect to $\bqi{i}$. We now extend the bounds to $\bqi{i+1}$, for which we need to remove the correctly priced items corresponding to bidders in $N_i$. 
\begin{claim}\label{clm:qi-qi+1}
	We have $\expect{\bqi{i+1}(\Ci{i+1})} \geq \expect{\bqi{i}(\Ci{i+1})} - \frac{\OPT}{10\beta}$.
\end{claim}
\begin{proof}
	For a bidder $b$, we write $\Ci{i}_b$ as the set of items in $\Ci{i}$ that are allocated to $b$ in $\Ostar$ (i.e., take their price in $\bqstar$ because of bidder $b$); this is similarly defined for $\Ci{i+1}_b$. 
	We can write, 
	\begin{align*}
		\expect{\bqi{i+1}(\Ci{i+1})} &~=~ \expect{\bqi{i}(\Ci{i+1}) - \sum_{b \in N_i} \bqi{i}(\Ci{i+1}_b)} ~ \geq ~ \expect{\bqi{i}(\Ci{i+1}) - \sum_{b \in N_i} \bqi{i}(\Ci{i}_b)}, 
		\intertext{because $\Ci{i+1} \subseteq \Ci{i}$. Since each bidder joins $N_i$ with probability $(1/10\beta)$, this implies } 
				\expect{\bqi{i+1}(\Ci{i+1})} &~\geq~ \expect{\bqi{i}(\Ci{i+1})} - \frac{1}{10\beta} \cdot \bqi{i}(\Ci{i}_b) ~ \geq ~ \expect{\bqi{i}(\Ci{i+1})} - \frac{\OPT}{10\beta}.  \Qed{Claim~\ref{clm:qi-qi+1}} 
	\end{align*}
	
\end{proof}

We now have all the ingredients needed to prove Lemma~\ref{lem:main}. 

\begin{proof}[Proof of Lemma~\ref{lem:main}]
	By applying Lemma~\ref{lem:random-arrival} to every $j \in [\alpha]$, we have 
	\begin{align}
		20\beta \cdot \sum_{j=1}^{\alpha} \expect{\val{\Ai{i}_j}} + \expect{\sum_{j=1}^{\alpha}\bqi{i}(\Ai{i}_j \cap \Di{i}_j)} \geq \sum_{j=1}^{\alpha}\bqi{i}(\Di{i}_j). \label{eq:lhs-rhs}
	\end{align}
	The RHS of above is clearly $\bqi{i}(\Ci{i})$. The second term in the LHS can be upper bounded by Lemma~\ref{lem:no-overestimate} and Claim~\ref{clm:qi-qi+1}, 
	\begin{align*}
		\expect{\sum_{j=1}^{\alpha}\bqi{i}(\Ai{i}_j \cap \Di{i}_j)} \quad \leq \quad \expect{\bqi{i}(\Ci{i+1})} +  \frac{\OPT}{10\beta} 
		\quad \leq \quad  \expect{\bqi{i+1}(\Ci{i+1})} + \frac{2\cdot\OPT}{10\beta}.
	\end{align*}
	Plugging in these bounds in Eq~(\ref{eq:lhs-rhs}), we obtain
	\begin{align}
		20\beta \cdot \sum_{j=1}^{\alpha} \expect{\val{\Ai{i}_j}} + \expect{\bqi{i+1}(\Ci{i+1})}  \geq \bqi{i}(\Ci{i}) - \frac{\OPT}{5\beta} \label{eq:cases}.
	\end{align}
	
	Now let us consider two cases. First suppose, 
	\begin{align}
	\sum_{j=1}^{\alpha} \expect{\val{\Ai{i}_j}} \geq \frac{\OPT}{200\beta^2}. \label{eq:cont}
	\end{align}
	In this case, $\expect{\val{\Ai{i}_{\jstar}}}$ for $\jstar$ chosen uniformly at random from $[\alpha]$ is at least $\frac{\OPT}{100\alpha\beta^2}$, hence satisfying item~\ref{item:main-allocatable} of the lemma (Allocatable case). 
	We now consider the other case where the LHS of Eq~\eqref{eq:cont} is smaller than the RHS. Plugging in this bound in Eq~\eqref{eq:cases} implies that
	\begin{align*}
		\expect{\bqi{i+1}(\Ci{i+1})} \quad \geq \quad \bqi{i}(\Ci{i}) - \frac{\OPT}{5\beta} - 20\beta \cdot \frac{\OPT}{200\beta^2} \quad >\quad \bqi{i}(\Ci{i}) - \frac{\OPT}{3\beta}. 
	\end{align*}
	This satisfies item~\ref{item:main-learnable} of the lemma (Learnable case),  concluding the proof. \Qed{Lemma~\ref{lem:main}}
	
\end{proof}


\newcommand{\Alg}{\ensuremath{Alg}}
\newcommand{\Algi}[1]{\ensuremath{\Alg^{(#1)}}}

\renewcommand{\alg}{\ensuremath{\textnormal{\textsf{ALG}}}}
\newcommand{\algi}[1]{\ensuremath{\alg^{(#1)}}}

\newcommand{\optstar}{\opt^{\star}}

\subsection{Proof of Theorem~\ref{thm:main-mech} -- Approximation Ratio}\label{sec:thm-main-mech}

We now prove the bound on expected approximation ratio of $\mech$. We first need some definitions. 

For an iteration $i \in [\beta+1]$, we use $\Algi{i} = (\Alg_1,\ldots,\Alg_n)$ to denote the allocation returned by our mechanism, conditioned on the mechanism reaching iteration $i$ and on the outcomes of 
$N_1,\ldots,N_{i-1}$ as well as the choice of $\TTstar$. We use $\algi{i}$ to denote the welfare of allocation $\Algi{i}$. We note that except for $i=\beta+1$, $\Algi{i}$ is a random variable. 

Our main tool in this section is the following inductive lemma.

\begin{lemma}\label{lem:induction}
For $i\in [\beta+1]$, 
\begin{align*}
\expect{\algi{i}} \geq \frac{1}{200 \alpha \beta^3} \cdot \left(1-\frac{1}{\beta}\right)^{\beta+1-i} \cdot \paren{ \bqi{i}(\Ci{i})- \frac{\OPT}{3\beta}(\beta + 1-i)},
\end{align*}
where the expectation is taken over the choice of $N_i,N_{i+1},\ldots,N_{\beta}$.
\end{lemma}

Before proving this lemma, we show how it immediately implies the proof of Theorem~\ref{thm:main-mech}. 
\begin{proof}[Proof of Theorem~\ref{thm:main-mech} -- Approximation Ratio]

By Lemma~\ref{lem:induction} for $i=1$,  
\begin{align*}
\expect{\algi{1}} ~ \geq ~ \frac{1}{200 \alpha \beta^3} \cdot \left(1-\frac{1}{\beta}\right)^{\beta} \cdot \paren{\bqi{1}(\Ci{1}) - \frac{\OPT}{3\beta} \cdot \beta} ~ = ~ \Omega\Big(\frac{1}{\alpha\beta^3}\Big) \cdot \paren{\bqi{1}(\Ci{1}) - \frac{\OPT}{3}}. 
\end{align*}
The only random event that we have not conditioned on in $\bqi{1}(\Ci{1})$ is the choice of $\TTstar$. Let $\alg$ denote the welfare of allocation returned by the mechanism. We have, 
\begin{align*}
	\expect{\alg} ~ = ~ \Omega\Big(\frac{1}{\alpha\beta^3}\Big) \cdot \Exp_{\TTstar}\bracket{\bqi{1}(\Ci{1}) - \frac{\OPT}{3}} ~ = ~ \Omega\Big(\frac{1}{\alpha\beta^3}\Big) \cdot \paren{\frac{\OPT}{2} - \frac{\OPT}{3}} ~ = ~ \Omega\Big(\frac{1}{\alpha\beta^3}\Big) \cdot \OPT, 
\end{align*}
where the second equality is because $\TTstar$ is chosen uniformly at random to be $\TTodd$ or $\TTeven$ and the bins in these two price trees partition the prices in $\bq(O)$ by Assumption~\ref{assumption1}.  
As $\alpha=\Theta(1)$, and $\beta = O(\log\log{\ratioprice})$, which is $O(\log\log{m})$ under Assumption~\ref{assumption1}, 
we obtain that $\mech$ achieves an $O((\log\log{m})^3)$ approximation in expectation. \Qed{Theorem~\ref{thm:main-mech}}

\end{proof}

We prove Lemma~\ref{lem:induction} using backward induction. We first show that the lemma easily holds true for the base case, namely, for $i=\beta+1$, because of the performance of \spmech for correctly priced items (Lemma~\ref{lem:fixed-price}).
The heart of the induction step lies in Learnable-Or-Allocatable Lemma (Lemma~\ref{lem:main}) that states $\expect{\bqi{i+1}(\Ci{i+1})}$ is close to  $\bqi{i}(\Ci{i})$ unless we already have a good allocation. So we first use the induction hypothesis to show that the 
expected welfare of the mechanism is close to $\expect{\bqi{i+1}(\Ci{i+1})}$ and then use Lemma~\ref{lem:main} to show it is close to $\bqi{i}(\Ci{i})$.

\begin{proof}[Proof of Lemma~\ref{lem:induction}]
We use backward induction on $i$. Consider the base case for $i=\beta+1$,  where we want to show the following (note that $\algi{\beta+1}$ is no longer a random variable)
\[ 
{{\algi{\beta+1}}} \geq \frac{1}{200 \alpha \beta^3}  \cdot \left( \bqi{\beta+1}(\Ci{\beta+1}) \right).
\]
Since our mechanism has already reached iteration $i=\beta+1$, this means that for every correctly priced item $j \in \Ci{\beta+1}$, $p_j \in \bpricei{\beta+1}$ and $q_j \in \bqi{\beta+1}$ both belong to a leaf-node of $\TTstar$, and consequently the same 
price bin. As such, by construction of bins, $p_j \leq q_j \leq \gamma \cdot p_j$ and hence running $\spmech(N_{\beta+1},M,\bpricei{i+1}/2)$ in this step of $\mech$, by Lemma~\ref{lem:fixed-price}, results in allocation with welfare, 
\begin{align*}
	\algi{\beta+1} \geq \frac{1}{\gamma} \cdot \bqi{\beta+1}(\Ci{\beta+1}) > \frac{1}{200\alpha\beta^3} \cdot \bqi{\beta+1}(\Ci{\beta+1}),
\end{align*}
by the choice of $\gamma = \Theta(\alpha\beta)$ in Eq~(\ref{eq:equations}). This proves the induction base. 

We now prove the induction step. Suppose the lemma is true for iterations $\geq i+1$ and we prove the induction step for iteration $i$. Notice that  w.p. $1/\beta$ the mechanism outputs an allocation $\Ai{i}_{\jstar}$ for $\jstar$ chosen randomly from $[\alpha]$, and otherwise it continues to the next iteration. This implies:
\begin{align} \label{eq:wrapUp}
\expect{\algi{i}} &\geq \frac{1}{\beta} \Exp_{N_i,\jstar}\bracket{\val{\Ai{i}_{\jstar}}} + \paren{1-\frac{1}{\beta}} \cdot  \Exp_{N_i}\bracket{{\algi{i+1}}} \notag \\
  &\geq  \frac{1}{\beta} \cdot \Exp_{N_i,\jstar}\bracket{\val{\Ai{i}_{\jstar}}} + \left(1-\frac{1}{\beta}\right)^{\beta+1-i}   \Exp_{N_i}\left[\frac{1}{200 \alpha \beta^3} \left( \bqi{i+1}(\Ci{i+1}) - \frac{\OPT}{3\beta}(\beta -i) \right)\right],
 \end{align}
where the second inequality uses  induction hypothesis. Now to prove the induction step we consider the two cases corresponding to Lemma~\ref{lem:main}:
 \begin{enumerate}[label=(\roman*)]
\item \textbf{Learnable case}, i.e., $\expect{\bqi{i+1}(\Ci{i+1})} \geq \bqi{i}(\Ci{i}) - \frac{\opt}{3\beta}$:  Combining this with Eq~\eqref{eq:wrapUp}, 
\begin{align*} 
\expect{\algi{i}}  ~\geq~  \left(1-\frac{1}{\beta}\right)^{\beta+1-i} \cdot  \Exp_{N_i}\bracket{\frac{1}{200 \alpha \beta^3} \paren{ \bqi{i}(\Ci{i}) - \frac{\OPT}{3\beta}- \frac{\OPT}{3\beta}(\beta -i) }},
\end{align*}
which implies the induction step.

\item \textbf{Allocatable case}, i.e., $\expect{\val{\Ai{i}_{\jstar}}} \geq \frac{\OPT}{200\alpha \cdot \beta^2}$: Combining this with Eq.~\eqref{eq:wrapUp}, 
\[ \expect{\algi{i}} ~ \geq ~ \frac{\OPT}{200\alpha \beta^3} ~ \geq ~  \frac{1}{200 \alpha \beta^3}  \left(1-\frac{1}{\beta}\right)^{\beta+1-i}  \expect{\bqi{i}(\Ci{i}) - \frac{\OPT}{2\beta}(\beta + 1-i)},
\]
where the last inequality uses $\OPT \geq \bqi{i}(\Ci{i})$ and implies the induction step.
\end{enumerate}
This concludes the proof of the lemma. \Qed{Lemma~\ref{lem:induction}}

\end{proof}


\newcommand{\Nstat}{\ensuremath{N_{\textnormal{\textsf{stat}}}}}
\newcommand{\OPTstat}{\ensuremath{\OPT_{\textnormal{\textsf{stat}}}}}
\newcommand{\ALGstat}{\ensuremath{\alg_{\textnormal{\textsf{stat}}}}}
\newcommand{\Nmech}{\ensuremath{N_{\textnormal{\textsf{mech}}}}}
\newcommand{\OPTmech}{\ensuremath{\OPT_{\textnormal{\textsf{mech}}}}}
\newcommand{\ALGmech}{\ensuremath{\alg_{\textnormal{\textsf{mech}}}}}

\newcommand{\fmech}{\ensuremath{\textnormal{\textsf{FinalMechanism}}}\xspace}

\section{Removing the Extra Assumptions}\label{sec:end-mech}

We now show how to remove Assumption~\ref{assumption1} and prove our main result in its full generality. 
We shall emphasize that the main contribution of our work is in establishing Theorem~\ref{thm:main-mech}; the remaining ideas here are standard for the most part and appear in similar forms in previous work on truthful mechanisms, e.g. in~\cite{DobzinskiNS06,Dobzinski07,Dobzinski16}. We  present them for completeness.

Let $O=(O_1,\ldots,O_n)$ be an optimal allocation with welfare $\OPT$ and supporting prices $\bq$. 
In order to remove Assumption~\ref{assumption1}, we find prices $\minprice$ and $\maxprice$ such that $\maxprice/\minprice = O(m^2)$, and for most items allocated by $O$, their price in $\bq$ belongs to the range $[\minprice : \maxprice]$; here, ``most 
items'' should be interpreted as items with prices in $\bq$ that is a constant fraction of $\OPT$. Having found such prices, we can then run $\mech$ from Section~\ref{sec:mech} and apply
Theorem~\ref{thm:main-mech} to finalize the proof (strictly speaking, Assumption~\ref{assumption1} stated that \emph{all} prices in all valuations of bidders need to be in range $[\minprice : \maxprice]$; however, as is evident from the proof of 
Theorem~\ref{thm:main-mech}, we only applied this assumption to prices in $\bq$). 

To find $\minprice$ and $\maxprice$, we partition $N$ into two (almost) equal-size groups $\Nstat$ and $\Nmech$ randomly. We run any constant-factor approximation algorithm (and not a truthful mechanism) for welfare maximization
with bidders in $\Nstat$ and items $M$, say, the algorithm of~\cite{LehmannLN06}, to compute a value $\ALGstat$ which is an $O(1)$-approximation to $\OPTstat$ namely, the value of welfare maximizing allocation for $\Nstat$ and $M$. 
We completely ignore the allocation of these bidders and instead only set $\minprice := \ALGstat/m^2$ and $\maxprice := \ALGstat \cdot \Theta(1)$. We then run $\mech(\Nmech,M)$ with $\minprice$ and $\maxprice$, and return
the resulting allocation to bidders in $\Nmech$.  

The intuition behind the approach is that because we partitioned $N$ into two \emph{random} groups, $\OPTstat$ and consequently $\ALGstat$ should be an $O(1)$-approximation to $\OPTmech$, namely, the value of welfare optimizing allocation for bidders in
$\Nmech$ (this intuition is not quite correct but for the moment let us ignore this fact). Thus, in an optimal allocation of $M$ to $\Nmech$, no item have price more than $\maxprice$ and also the total contribution of items with price smaller than $\minprice$ is 
negligible, hence we can safely ignore them. This in turn implies that Assumption~\ref{assumption1} holds and by Theorem~\ref{thm:main-mech}, $\mech(\Nmech,M)$ outputs an allocation with
welfare $\ALGmech$ such that $\expect{\ALGmech} \geq \OPTmech \cdot \Omega\paren{\frac{1}{(\log\log{m})^3}}$. Moreover, by the choice of $\Nmech$, we have $\expect{\OPTmech} = \OPT/2$. Thus, this should gives us 
an $O((\log\log{m})^3)$ approximation in expectation. 

As stated earlier, there is a slight problem with the above intuition. One cannot in general guarantee that by partitioning the bidders into two parts randomly, each part will have roughly the same contribution to the value of $\OPT$. In particular, if there exists a 
bidder with a much higher contribution to $\OPT$ than the rest, the above approach is bound to fail. So we take care of this case separately as follows: With probability half, we simply run a second-price auction on the grand bundle $M$ and 
sell it to the highest bidder entirely. With the remaining half probability, we run the above procedure. This ensures that if such a bidder exists, we get her contribution with probability half. Otherwise, with probability half, we can run the previous analysis. 

\subsection{The Final Mechanism}\label{sec:final}

Our final mechanism is as follows. 

\begin{tbox}
	\underline{$\fmech(N,M)$}
	 \begin{enumerate}
	 	\item With probability $1/2$, run a second-price auction on grand bundle $M$ with all bidders, return the resulting allocation, and terminate. With the remaining probability, continue. 
		\item Pick $\Nstat$ by sampling each bidder in $N$ independently and w.p. $1/2$. Let $\Nmech := N \setminus \Nstat$. 
		\item\label{line:alg} Run the $2$-approximation algorithm of~\cite{LehmannLN06} on items $M$ and bidders $\Nstat$. Let $\ALGstat$ be the welfare of the returned allocation. Let $\minprice := \ALGstat/m^2$ and $\maxprice := 8 \cdot \ALGstat$. 
		\item\label{line:mech} Run $\mech(\Nmech,M)$ with $\minprice$ and $\maxprice$, and return the allocation.  
	 \end{enumerate}
\end{tbox}

We have the following theorem that formalizes our main result from Section~\ref{sec:intro}. 
\begin{theorem}\label{thm:final}
	For a combinatorial auction with $n$ submodular (even XOS) bidders and $m$ items, $\fmech$ is universally truthful, uses $\poly(m,n)$ demand and value queries, and achieves an approximation ratio of $O((\log\log{m})^3)$ in expectation. 
\end{theorem}

The proof of truthfulness of Theorem~\ref{thm:final} is quite easy. The case where we run the second-price auction is clearly truthful. For the other case, note that we never allocate any item to bidders in $\Nstat$ and so they might as well reveal their
true valuations in response to the algorithm in Line~\eqref{line:alg}. Finally, $\mech$ with bidders $\Nmech$ is truthful by Theorem~\ref{thm:main-mech}. The bound on the number of queries also follows from~\cite{LehmannLN06} for Line~\eqref{line:alg} 
and Theorem~\ref{thm:main-mech} for Line~\eqref{line:mech}, and since the second-price auction can be implemented with $n$ value queries for the grand bundle. It thus only remains to analyze the approximation ratio of $\fmech$, which we do in the next section. 

\subsection{Approximation Ratio of Final Mechanism}\label{sec:final-analysis}

We use the following standard result that follows directly from  Chernoff-Hoeffding bound.

\begin{lemma}[cf.~\cite{DobzinskiNS06,Dobzinski07,Dobzinski16}]\label{lem:chernoff-bidder}
	Let $O=(O_1,\ldots,O_n)$ be an optimal allocation of items $M$ to bidders $N$ with welfare $\OPT$. Suppose we sample each $i \in N$ w.p. $\rho$ independently to obtain $N'$. If for every $i \in N$, we have
	$v_i(O_i) \leq \epsilon \cdot \OPT$, then $\sum_{i \in N'} v_i(O_i) \geq (\rho/2) \cdot \OPT$ w.p. at least $1-2\cdot\exp\paren{-\frac{\rho}{2 \cdot \epsilon}}$. 
\end{lemma}


Fix an optimal allocation $O=(O_1,\ldots,O_n)$ of items to bidders in $N$ with welfare $\OPT$. We say that a bidder $i \in N$ is \emph{dominant} iff $v_i(O_i) \geq \OPT/8$. For the analysis, we consider two cases: either $(i)$ there exists at least one dominant bidder, or $(ii)$ no bidder is dominant.  

\paragraph{Case $(i)$: A dominant bidder exists.} W.p. half, we decide to run the second-price auction. Let $i$ be the bidder that gets the grand bundle $M$ in the auction. Clearly $i$ has to be a dominant bidder in this case and 
thus $v_i(M) \geq \OPT/8$ already. As such, in this case, the expected welfare of the allocation is at least $\OPT/16$, concluding the proof. 

\paragraph{Case $(i)$: No dominant bidder exists.} W.p. half, we decide not to run the second-price auction. Let $\OPTstat$ and $\OPTmech$ be the welfare of the optimal allocation of $M$ to $\Nstat$ and $\Nmech$, respectively. 
By Lemma~\ref{lem:chernoff-bidder}, applied to choice of $\Nstat$ and $N \setminus \Nstat$ (both sets have the same distribution) with  $\rho = 1/2$ and $\eps = 1/8$, and a union bound, w.p. at least $1/2$, we have
\begin{align}
	\frac{1}{4} \cdot \OPT ~\leq ~ \OPTstat ~ \leq ~ \OPT \qquad \textnormal{and} \qquad \frac{1}{4} \cdot \OPT ~ \leq ~ \OPTmech ~ \leq ~  \OPT. \label{eq:opt-stat-mech}
\end{align}
In the following, we condition on the (independent) events that we do not run the second-price auction, and that Eq~\eqref{eq:opt-stat-mech} holds, which happens w.p. $1/4$.  

Fix a welfare maximizing allocation of $M$ to $\Nmech$ with welfare $\OPTmech$ and supporting prices $\bq = (q_1,\ldots,q_m)$. 
Since we run a $2$-approximation algorithm in Line~\eqref{line:alg}, we know that $\frac{1}{8} \cdot \OPT \leq \ALGstat \leq \OPT$ by Eq~\eqref{eq:opt-stat-mech}. Hence, setting $\maxprice=8 \cdot \ALGstat$ ensures
that $q_j \leq \maxprice$ for every item $j \in M$. Moreover, let $M' \subseteq M$ be the set of   items $j$ such that $q_j \leq \minprice$. By definition of $\minprice = \ALGstat/m^2$ and since $\ALGstat > \OPTmech/8$, we get
$\bq(M') \leq \OPTmech/2$ (as $m \gg 8$). As such, we can simply ignore the contribution of all items in $M'$ and still have a set of items $M \setminus M'$ that can be allocated to bidders in $\Nmech$ with welfare at least $\OPTmech/2 \geq \OPT/8$. Moreover, the supporting prices of these items now belong to $[\minprice:\maxprice]$. Hence, we can apply Theorem~\ref{thm:main-mech} under Assumption~\ref{assumption1} and
obtain that in this case, the expected welfare of the allocation is at most $O((\log\log{m})^3)$ times smaller than $\OPT$, finishing the proof of Theorem~\ref{thm:final}.


\section{Concluding Remarks and Open Problems}\label{sec:conc}

We gave a randomized, computationally-efficient, and universally truthful mechanism for combinatorial auctions with submodular (even XOS) bidders that achieves an $O((\log\log{m})^3)$-approximation.  
This reduces the gap between the approximation ratio achievable by truthful mechanisms vs arbitrary algorithms for this problem by an exponential factor from $\poly{(\log{(m)})}$ to $\poly{(\log\log{(m)})}$.  

The obvious question left open by our work is whether this gap  can be  improved further.  We do not believe in any way that our $O((\log\log{m})^3)$ approximation is the best possible\footnote{Indeed, using a slightly more nuanced argument, our bounds can be 
improved to $O(\frac{(\log\log{m})^3}{\log\log\log{m}})$; however as this $\Theta(\log\log\log{m})$ improvement is  minor and for the sake of clarity, we used the slightly weaker analysis in the paper.}. 
On the other hand, the limit of our approach seems to be an $\Omega(\log\log{m})$ approximation. It is a fascinating open question whether one can improve the approximation factor all the way down to a constant. However, 
even improving the approximation ratio of our mechanism down to $O(\log\log{m})$ already seems challenging, and is an interesting open question. On the lower bound front, proving any separation between the power of truthful mechanisms and algorithms when 
the access to input is via arbitrary queries, namely, the communication complexity setting, is also very interesting.

\subsection*{Acknowledgements}

We are grateful to Matt Weinberg for illuminating discussions on the related work, and the anonymous reviewers of FOCS 2019 for many helpful comments on the presentation of this paper.

{\small
\bibliographystyle{abbrv}
\bibliography{new}
}

\clearpage
\appendix


\renewcommand{\Ms}{M^{\star}}

\section{Missing Details}\label{app:missing}

\subsection{Formal Definitions of Mechanisms and Truthfulness}\label{app:truthful} 
 
Let $\VV$ be a class of valuation functions defined over $M$, say, all submodular functions $2^M \rightarrow \IR^+$, and $\AA$ be the set of all possible allocations of $M$ to $n$ bidders.  
A deterministic mechanism for combinatorial auctions is a pair $(f,\bprice)$ where $f: \VV^n \rightarrow \AA$ (representing the allocation to bidders) 
and $\bprice=(p_1,\ldots,p_n)$ where $p_i : \VV^n \rightarrow \IR^+$ (representing the price charged for item $i$). A randomized mechanism is simply a probability distribution over deterministic mechanisms. 

\begin{definition}[Truthfulness and Universal Truthfulness]\label{def:truthful}
	A deterministic mechanism $(f,p)$ is \emph{truthful} iff for all $i \in N$, $v_i,v'_i \in \VV$ and $v_{-i} \in \VV^{n-1}$, we have, 
	\[
		v_i(f(v_i,v_{-i})_i) - p_i(v_i,v_{-i}) \geq v_i(f(v'_i,v_{-i})_i) - p_i(v'_i,v_{-i}).
	\]
	A randomized mechanism is \emph{universally truthful} iff it is a distribution over truthful mechanisms. 
\end{definition}

We note that beside universal truthfulness, the notion of truthful-in-expectation is also considered for randomized mechanisms that guarantee that bidding truthfully maximizes the \emph{expected} profit; see, e.g.~\cite{LaviS05, DughmiV11,DughmiRY11} and references therein.  
This is a much weaker guarantee than universal truthfulness we consider in this paper. In particular such mechanisms are only applicable when bidders are  risk neutral and have no information about the outcomes of the random coin flips before they need to act;
see~\cite[Section 1.2]{DobzinskiNS06} for more details. 


\subsection{Proof of Lemma~\ref{lem:fixed-price} -- Fixed-Price Auctions}\label{app:lem-fixed-price}

\begin{lemma*}[Restatement of Lemma~\ref{lem:fixed-price}]
	Let $A:=\spmech(N,M,\bprice)$ and $\delta <1/2$ be a parameter. 
	Suppose $O$ is any allocation with supporting prices $\bq$ and $\Ms \subseteq M$ is the set of items $j$ with $\delta \cdot q_j \leq p_j < \frac{1}{2} \cdot q_j$. 
	Then, $\val{A} \geq \delta \cdot \bq(\Ms)$. 
\end{lemma*}
\begin{proof}
	Define the allocation $\Ostar = (\Ostar_1,\ldots,\Ostar_n)$ as the restriction of $O$ to $\Ms$. Define $\barA_i = \Ostar_i \setminus A$ for every $i \in N$ and $\barA := \barA_1 \cup \ldots \cup \barA_n$. 
	Bidder $i$ could have chosen $\barA_i$ in \spmech but decided to pick another bundle $A_i$ instead. This implies that: 
	\begin{align}
	v_i(A_i) - \bprice(A_i) \geq v_i(\barA_i) - \bprice(\barA_i). \label{eq:fixed-price-1}
	\end{align}
	We now use Eq~(\ref{eq:fixed-price-1}) to prove the lemma. We have, 
	\begin{align*}
		\val{A} &= \sum_{i=1}^{n} v_i(A_i) \quad = \quad \bprice(A) + \sum_{i=1}^{n} (v_i(A_i) - \bprice(A_i))   \\
		&\geq \bprice(A) + \sum_{i=1}^{n} (v_i(\barA_i) - \bprice(\barA_i)) \quad \geq \quad  \bprice(A) + \sum_{i=1}^{n} (\bq(\barA_i) - \bprice(\barA_i))
	\end{align*}
	by Eq~(\ref{eq:fixed-price-1}) and since $\barA_i \subseteq \Ostar_i \subseteq O_i$. Now using as $p_j \leq q_j/2$ for all $j \in \barA \subseteq \Ostar = \Ms$, we get
		\begin{align*}
			\val{A}	&\geq \bprice(A) + \sum_{i=1}^{n} \bprice(\barA_i)  \\
		&\geq \bprice(\Ostar) \quad  \geq  \quad \delta \cdot \bq(\Ostar),
	\end{align*}
where the last two  inequalities use $\Ostar \subseteq \barA \cup A$ and  $\barA \cap A = \emptyset$, and that $p_j \geq \delta \cdot q_j$ for all $j \in \Ostar$.
	This concludes the proof as $\bq(\Ostar) = \bq(\Ms)$ by definition. 
\end{proof}


\end{document}